\documentclass[a4paper,UKenglish,cleveref, autoref, thm-restate]{lipics-v2021}

\nolinenumbers
\usepackage{mathrsfs}
\usepackage{amsthm,amsmath}
\usepackage{amstext}
\usepackage{boxedminipage}
\usepackage[linesnumbered, boxed, algosection,noline]{algorithm2e}
\usepackage{graphics}
\usepackage{framed}
\usepackage{thm-restate}
\usepackage{thmtools}
\usepackage{tcolorbox}

\newcommand{\OO}{\mathcal{O}}

\newcommand{\yes}{\textsf{Yes}}
\newcommand{\no}{\textsf{No}}

\newcommand{\NP}{\textsf{NP}}
\newcommand{\coNP}{\textsf{co-NP}}

\newcommand{\fvstfull}{{\sc Feedback Vertex Set in Tournaments}\xspace}
\newcommand{\fvst}{{\sc FVST}\xspace}

\newcommand{\dhsfull}[1]{$#1$-\text{{\sc Hitting Set}}\xspace}

\newcommand{\cvdfull}{{\sc Cluster Vertex Deletion}\xspace}
\newcommand{\cvd}{{\sc CVD}\xspace}

\DeclareMathOperator{\poly}{poly}
\newcommand{\good}{output-parameter sensitive}

 \title{Lossy Kernelization for (Implicit) Hitting Set Problems}

\author{Fedor V. Fomin}{University of Bergen, Bergen, Norway.}{ fomin@ii.uib.no}{}{Research Council of Norway via the project BWCA (grant no. 314528).}
\author{Tien-Nam Le}{\'{E}cole Normale Sup\'{e}rieure de Lyon, Lyon, France.}{tien-nam.le@ens-lyon.fr}{}{}
\author{Daniel Lokshtanov}{University of California Santa Barbara, USA.}{daniello@ucsb.edu}{}{Supported by NSF award CCF-2008838.}
\author{Saket Saurabh}{The Institute of Mathematical Sciences, HBNI, Chennai, India, and University of Bergen, Bergen, Norway.} {saket@imsc.res.in}{}{European Research Council (ERC) grant agreement no.~819416, and
Swarnajayanti Fellowship no.~DST/SJF/MSA01/2017-18.}
\author{St\'{e}phan Thomass\'{e}}{\'{E}cole Normale Sup\'{e}rieure de Lyon, Lyon, France.}{stephan.thomasse@ens-lyon.fr}{}{ANR projects TWIN-WIDTH (CE48-0014-01) and DIGRAPHS (CE48-0013-01).}
\author{Meirav Zehavi}{Ben-Gurion University of the Negev, Beersheba, Israel.}{zehavimeirav@gmail.com}{}{European Research Council (ERC) grant titled PARAPATH.}

\authorrunning{F.~V.~Fomin, T.~Le, D.~Lokshtanov, S.~Saurabh, S.~Thomass\'{e} and M.~Zehavi}

\Copyright{Fedor V.~Fomin, Tien-Nam Le, Daniel Lokshtanov, Saket Saurabh, St\'{e}phan Thomass\'{e} and Meirav Zehavi}



\EventEditors{John Q. Open and Joan R. Access}
\EventNoEds{2}
\EventLongTitle{42nd Conference on Very Important Topics (CVIT 2016)}
\EventShortTitle{CVIT 2016}
\EventAcronym{CVIT}
\EventYear{2016}
\EventDate{December 24--27, 2016}
\EventLocation{Little Whinging, United Kingdom}
\EventLogo{}
\SeriesVolume{42}
\ArticleNo{23}

\begin{document}

\maketitle

\begin{abstract} 
We re-visit the complexity of polynomial time pre-processing (kernelization) for the {\sc $d$-Hitting Set} problem. This is one of the most classic  problems in Parameterized Complexity by itself, and, furthermore, it encompasses several other of the most well-studied problems in this field, such as {\sc Vertex Cover}, \fvstfull (\fvst) and \cvdfull (\cvd). In fact, {\sc $d$-Hitting Set} encompasses any deletion problem to a hereditary property that can be characterized by a finite set of forbidden induced  subgraphs. With respect to bit size, the kernelization complexity of {\sc $d$-Hitting Set} is essentially settled: there exists a kernel with  $\OO(k^d)$ bits ($\OO(k^d)$ sets and $\OO(k^{d-1})$ elements)   and this it tight by the result of Dell and van Melkebeek~[STOC 2010, JACM 2014]. Still, the question of whether there exists a kernel for  {\sc $d$-Hitting Set} with {\em fewer elements} has remained one of the most major open problems~in~Kernelization. 

In this paper, we first show that if we allow the kernelization to be {\em lossy} with a qualitatively better loss than the best possible approximation ratio of polynomial time approximation algorithms, then one can obtain kernels where the number of elements is linear for every fixed $d$.  Further, based on this, we present our main result: we show that there exist approximate Turing kernelizations for {\sc $d$-Hitting Set} that even beat the established bit-size lower bounds for exact kernelizations---in fact, we use a {\em constant} number of oracle calls, each with {\em ``near linear''} ($\OO(k^{1+\epsilon})$) bit size, that is, almost the best one could hope for. Lastly, for two special cases of implicit {\sc 3-Hitting set}, namely, \fvst and \cvd, we obtain the ``best of both worlds'' type of results---$(1+\epsilon)$-approximate kernelizations with a linear number of vertices. In terms of size, this substantially improves the exact kernels of Fomin et al.~[SODA 2018, TALG 2019], with simpler~arguments.  
\end{abstract}

\newpage

\newcommand{\cO}{\mathcal{O}}
\newcommand{\Oh}{\cO}
\newcommand{\eps}{{\varepsilon}}

\newcommand{\compass}{\NP\subseteq \coNP/\poly}
\newcommand{\ncompass}{\NP\nsubseteq \coNP/\poly}

\newcommand{\ncocompass}{\coNP\nsubseteq \NP/\poly}
 \newcommand{\conppoly}{\coNP/\poly} 
 \newcommand{\nppoly}{ \NP/\poly} 
 \newcommand{\nppolyfull}{\NP\subseteq\coNP/\poly} 
 \newcommand{\conppolyfull}{\coNP  \subseteq \NP/\poly} 


\newcommand{\ProblemFormat}[1]{{\sc #1}}

\newcommand{\ProblemIndex}[1]{\index{problem!\ProblemFormat{#1}}}

\newcommand{\ProblemName}[1]{\ProblemFormat{#1}\ProblemIndex{#1}\xspace}
\newcommand{\ProblemNameX}[2]{\ProblemFormat{#1}\ProblemIndex{#2}\xspace}
\newcommand{\probdHS}{\ProblemName{$d$-Hitting Set}}
\newcommand{\probTHS}{\ProblemName{$3$-Hitting Set}}
\newcommand{\probCVD}{\ProblemName{Cluster Vertex Deletion}}
\newcommand{\ProblemUberOpen}[2]{\parbox{\linewidth}{\noindent{}\colorbox{black!15}{\parbox{\linewidth}{#1}}\\[1mm]#2}\medskip}
\newcommand{\ProblemOpen}[3]{\ProblemUberOpen{#1}{\textbf{Input:}  #2\\\textbf{Question:}  #3}}

\section{Introduction}\label{sec:intro}
In \probdHS,  the
 input consists of a universe $U$, a family $\mathcal{F}$ of sets over $U$,
where each set in $\mathcal{F}$ is of size at most $d$, and an integer $k$. The task is to determine whether there exists a set  $S \subseteq U$, called a \emph{hitting set},  of size at most $k$ that has a nonempty
intersection with every set of $\mathcal{F}$.
The \probdHS problem is a classical optimization problem whose computational complexity has been studied for decades from the perspectives of different algorithmic paradigms.  Notably, \probdHS  is a generic problem, and hence, in particular, various computational problems can be re-cast in terms of it.  Of course,  {\sc Vertex Cover}, the most well-studied problem in Parameterized Complexity, is the special case of \probdHS with $d=2$.  More generally,  \probdHS encompasses a variety of (di)graph modification problems, where the task is to delete at most $k$ vertices (or edges) from a graph such that the resulting graph does not contain an induced subgraph (or a subgraph) from a family of forbidden  graphs  $\mathcal{F}$.
 Examples of some such well-studied problems include \cvdfull, {\sc $d$-Path Vertex Cover}, {\sc $d$-Component Order Connectivity},  {\sc $d$-Bounded-Degree Vertex Deletion}, {\sc Split Vertex Deletion} and    
 \fvstfull.

Kernelization, a subfield of Parameterized Complexity, provides a mathematical framework to capture  the performance of  polynomial time preprocessing. It makes it possible to quantify the degree to which polynomial time algorithms succeed at reducing input instances of \NP-hard problems. 
More formally,  every instance of a parameterized problem $\Pi$ is associated with an integer $k$,
which is called the {\em parameter}, and $\Pi$ is said to admit a {\em kernel} if there is a polynomial-time algorithm, called a {\em kernelization algorithm}, that reduces the input instance of $\Pi$ down to an equivalent instance of $\Pi$ whose size is bounded by a function $f(k)$ of $k$. (Here, two instances are equivalent if both of them are either \yes-instances or \no-instances.) Such an algorithm is called an {\em $f(k)$-kernel} for $\Pi$. If $f(k)$ is a polynomial function of $k$, then we say that the kernel is a \emph{polynomial kernel}.  Over the last decade, Kernelization has become a central and active field of study, which stands at the forefront of Parameterized Complexity, especially with the development of complexity-theoretic lower bound tools for kernelization. These tools can be used to show that a polynomial kernel~\cite{BDFH09,D15,FS11,KratschW12}, or a kernel of a specific size~\cite{DM12,DM14,HW12} for concrete problems would imply an unlikely complexity-theoretic collapse. 
We refer to the recent book on kernelization \cite{FominLSZ19} 
for a detailed treatment of the area of kernelization. In this paper, we provide a number of positive results on the kernelization complexity of \probdHS, as well as on several special cases of   \probTHS.

The most well-known example of a polynomial kernel, which, to the best of our knowledge, is taught in the first class/chapter on kernelization of any course/book that considers this subject,  is the classic kernel for {\sc Vertex Cover} ({\sc $2$-Hitting Set}) that is based on Buss rule. 
More generally, one of the most well-known examples of a polynomial kernel is a kernel with $\OO(k^d)$ sets and elements for \probdHS (when $d$ is a fixed constant) using the Erd\"{o}s-Rado Sunflower lemma.\footnote{The origins of this result are unclear. The first kernel with $\OO(k^d)$ sets appeared in 2004~\cite{FellowsKNRRSTW08},
 but the authors do not make use of the Sunflower Lemma. To the best of our knowledge, the first exposition of the kernel based on the Sunflower Lemma appears in the book of Flum and
 Grohe~\cite{FG06}.}
Complementing this positive result,  originally in 2010, a celebrated result by Dell and van Melkebeek~\cite{DM14} showed that unless $\conppolyfull$, for any $d\geq 2$ and any $\epsilon >0$,  \dhsfull{d} 
 does not admit a kernel with $\OO(k^{d-\epsilon})$ sets. 
 Hence, the kernel with $\OO(k^d)$ sets is essentially tight with respect to size. However, when it comes to the bound on the number of elements in a kernel, the situation is unclear. 
Abu-Khzam~\cite{Abu-Khzam10}  showed that   \probdHS  admits a kernel with at most $(2d - 1)k^{d-1} + k$ elements. However, we do not know whether this bound is tight or even close to that. As it was written in 
\cite[page~470]{FominLSZ19}: 
\begin{quote}
\emph{Could it be that  \probdHS admits a kernel with a polynomial in $k$ number of elements, where the degree of the polynomial does not depend on $d$?
 This does not look like a plausible conjecture, but we do not know how  to refute it either.}
 \end{quote}
 The origins of this question can be traced back to the open problems from  WorKer 2010~\cite[page 4]{BFS10-worker}. Moreover, in the list of open problems from WorKer 2013 and FPT School 2014~\cite[page 4]{cygan2014open}, the authors asked whether \probdHS  admits a kernel with $f(d)\cdot k$ elements for some function $f$ of $d$ only. After being explicitly stated at these venues, this question and its variants have been re-stated in a considerable number of papers~(see, e.g., \cite{DomGHNT10,FominLSZ19,YouW017,BessyFGPPST11}), and is being repeatedly asked in annual meetings centered around parameterized complexity. Arguably, this question has become the most major and longstanding open problem in kernelization for a specific problem.  In spite of many attempts, even for $d=3$,  the question whether \probdHS admits a kernel with $\OO(k^{2-\varepsilon})$ elements, for some $\epsilon >0$, has still remained open.

 From an approximation perspective, the optimization version of  \probdHS admits a trivial $d$-approximation. Up to the Unique Game Conjecture, this bound is tight---for any $\varepsilon>0$,  \probdHS does not admit a polynomial time $(d- \varepsilon)$-approximation~\cite{DBLP:journals/jcss/KhotR08}. So, at this front, the problem is essentially resolved.
 
With respect to kernelization, firstly, the barrier in terms of number of sets, and secondly, the lack of progress in terms of the number of elements, coupled with the likely impossibility of $(d- \varepsilon)$-approximation of \probdHS, bring lossy kernelization as a natural tool for further exploring of the complexity of this fundamental problem. We postpone the formal definition of  lossy kernelization to Section~\ref{sec:prelims}. Informally,   
a polynomial size $\alpha$-approximate kernel consists of two polynomial-time procedures. The first is a pre-processing algorithm that takes as input an instance $(I,k)$ to a parameterized problem, and outputs another instance $(I',k')$ to the same problem, such that $|I'|+k' \leq k^{\OO(1)}$. The second transforms, for every $c \geq 1$, a $c$-approximate solution $S'$ to the pre-processed instance $(I',k')$ into a $(c \cdot \alpha)$-approximate solution $S$ to the original instance $(I,k)$. 
Then, the main question(s) that we address in this paper is:
 \smallskip

\begin{tcolorbox}[colback=gray!5!white,colframe=gray!75!black]
Is it possible to obtain a lossy kernel for \probdHS 
with a qualitatively better loss than $d$ and with $\OO(k^{d-1-\varepsilon})$ bit-size, or at least  with $\OO(k^{d-1-\varepsilon})$ elements?  
\end{tcolorbox}
 
In this paper, we present a surprising answer: {\em not only the number of elements can be bounded by $\OO(k)$ (rather than just $\OO(k^{d-1-\varepsilon})$), but even the bit-size can ``almost'' be bounded by $\OO(k)$!} From the perspective of the size of the kernel, this is essentially the best that one could have hoped for. Still, we only slightly (though non-negligibly) improve on the approximation ratio $d$. For example, for $d=2$ ({\sc Vertex Cover}), we attain an approximation ratio of $1.721$. So, while we make a critical step that is also the first---in particular, we show that, conceptually, the combination of kernelization and approximation breaks their independent barriers---we also open up the door for further research of this kind, on this problem as well as other problems.

More precisely, we present the following results and concept.
We remark that for all of our results, we use an interesting fact about the natural Linear Programming (LP)  relaxation of 
 {\sc $d$-Hitting Set}: the support of any optimal LP solution to the LP-relaxation of {\sc $d$-Hitting Set} is of size at most $d\cdot\mathsf{frac}$ where $\mathsf{frac}$ is the optimum (minimum value) of the LP~\cite{furedi1988matchings}. Furthermore, to reduce bit-size rather than only element number, we introduce an ``adaptive sampling strategy'' that is, to the best of our knowledge, also novel in parameterized complexity. We believe that these ideas will find further applications in kernelization in the future.
 More information on our methods can be found in the next section.
 
\begin{itemize}
\item {\bf Starting Point: Linear-Element Lossy Kernel for {\sc $d$-Hitting Set}.} First, we show that {\sc $d$-Hitting Set} admits a $(d-\frac{d-1}{d})$-approximate $d\cdot\mathsf{opt}$-element kernel, where $\mathsf{opt}\leq k$ is the (unknown) optimum (that is, size of smallest solution).\footnote{In fact, when the parameter is $k$, we show that the bound is better.} For example, when $d=3$, the approximation ratio is $d-\frac{d-1}{d}=2\frac{1}{3}$, which is a notable improvement over $3$.  When $d=2$, this result encompasses the classic (exact) $2\cdot\mathsf{opt}$-vertex kernel for {\sc Vertex Cover}~\cite{chen2001vertex,nemhauser1974properties}. 
We also remark that our linear-element lossy kernel for {\sc $d$-Hitting Set} is a  critical component (used as a black box) in all of our other results. 

\item {\bf Conceptual Contribution: Lossy Kernelization Protocols.} We extend the notions of lossy kernelization and kernelization protocols\footnote{We remark that kernelization protocols are a highly restricted special case of Turing kernels, that yet generalizes kernels.} to {\em lossy kernelization protocols}. Roughly speaking, an $\alpha$-approximate kernelization protocol can perform a bounded in $k$ number of calls (called {\em rounds}) to an oracle that solves the problem on instances of size (called {\em call size}) bounded in $k$, and besides that it runs in polynomial time. Ideally, the number of calls is bounded by a fixed constant, in which case the protocol is called {\em pure}. Then, if the oracle outputs $c$-approximate solutions to the instances it is given, the protocol should output a $(c\cdot\alpha)$-approximate solution to the input instance. In particular, a lossy kernel is the special case of a lossy protocol with one oracle call. The {\em volume} of a lossy kernelization protocol is the sum of the sizes of the calls it performs.

\item {\bf Main Contribution: Near-Linear Volume and Pure Lossy Kernelization Protocol for {\sc $d$-Hitting Set}.} We remark that the work of Dell and van Melkebeek~\cite{DM14} further asserts that also the existence of an exact (i.e., $1$ approximate in our terms) kernelization protocol for {\sc $d$-Hitting Set} of volume $\OO(k^{d-\epsilon})$ is impossible unless $\conppolyfull$.

First, we show that {\sc Vertex Cover} admits a (randomized) 1.721-approximate kernelization protocol of $2$ rounds and call size $\OO(k^{1.5})$. This special case is of major interest in itself: {\sc Vertex Cover} is the most well-studied problem in Parameterized Complexity, and, until now, no result that breaks both bit-size and approximation ratio barriers simultaneously has been known.

Then, we build upon the ideas exemplified for the case of {\sc Vertex Cover} to significantly generalize the result: while {\sc Vertex Cover} corresponds to $d=2$, we are able to capture {\em all} choices of $d$. Thereby, we prove our main result: for any $\epsilon>0$,  {\sc $d$-Hitting Set} admits a (randomized) pure $(d-\delta)$-approximate kernelization protocol of call size $\OO(k^{1+\epsilon})$. Here, the number of rounds and $\delta$ are fixed constants that depend only on $d$ and $\epsilon$. While the improvement over the barrier of $d$ in terms of approximation is minor (though still notable when $d=2$), it is a {\em proof of concept}---that is, it asserts that $d$ is not an impassable barrier.\footnote{Possibly, building upon our work, further improvements on the approximation factor (though perhaps at the cost of an increase in the output size) may follow.} Moreover, it does so with almost the best possible (being almost linear) output~size. 

\item {\bf Outlook: Relation to Ruzsa-Szemer\'{e}di Graphs.} Lastly, we present a connection between the possible existence of a $(1+\epsilon)$-approximate kernelization protocol for {\sc Vertex Cover} of call size $\OO(k^{1.5})$ and volume $\OO(k^{1.5+o(1)})$ and a known open problem about Ruzsa-Szemer\'{e}di graphs (defined in Section \ref{sec:prelimsAppendix}). We discuss this result in more detail in Section~\ref{sec:overview}.
\end{itemize}

\bigskip
\noindent{\bf Kernels for Implicit $3$-Hitting Set Problems}.
Lastly, we provide better lossy kernels for two well-studied graph problems, namely,  \probCVD and \fvstfull, which are known to be implicit \dhsfull{3} problems~\cite{pcbook}. Notably, both our algorithms are based on some of the ideas and concepts that are part of our previous results, and, furthermore, we believe that the approach underlying the parts common to both these algorithms may be useful when dealing also with other hitting and packing problems of constant-sized objects. 
In the \probCVD problem, we are given 
a graph $G$ and an integer $k$. The task is to decide whether 
 there exists a set $S$ of at most $k$ vertices of $G$ such that $G-S$ is a cluster graph. Here, a cluster graph is a graph where every connected component is a clique. It is known that this problem can be formulated as a \dhsfull{3} problem where the family $\cal F$ contains the vertex sets of all {\em induced $P_3$'s} of $G$. (An induced $P_3$ is a path on three vertices where the first and last vertices are non-adjacent in $G$.)
 In the \fvstfull\ problem, we are given  a  {tournament} $G$ and an integer $k$. The task is to decide whether  there is a set $S$ of $k$ vertices such that each directed cycle of $G$ contains a member of $S$ (i.e., $G-S$ is acyclic). It is known that \fvstfull can be formulated as a \dhsfull{3} problem as well, where the family $\cal F$ contains the vertex sets of all directed cycles on three vertices (triangles) of $G$. 
 
 In  \cite{FominLLSTZ19}, it was shown that \fvstfull and \cvdfull  
admit kernels with $\OO(k^{\frac{3}{2}})$ vertices and  $\OO(k^{\frac{5}{3}})$ vertices, respectively. This answered an open question from WorKer 2010~\cite[page 4]{BFS10-worker}, regarding the existence of kernels with $\OO(k^{2-\epsilon})$ vertices for these problems. The question of the existence of linear-vertex kernels for these problems is open.
 In the realm of approximation algorithms, for \fvstfull, Cai , Deng and Zang~\cite{CaiDZ00} gave a factor $2.5$ approximation algorithm, which was later improved to $7/3$ by Mnich, Williams and Végh~\cite{MnichWV16}. Recently, Lokshtanov, Misra, Mukherjee, Panolan, Philip and Saurabh~\cite{DBLP:conf/soda/LokshtanovMMPP020} gave a $2$-approximation algorithm for \fvstfull.
 For \cvdfull, You, Wang and Cao~\cite{YouW017} gave a factor $2.5$ approximation algorithm, which later was improved to $7/3$ by Fiorini, Joret and Schaudt~\cite{FioriniJS16}. It is open whether \cvdfull admits a $2$-approximation algorithm. We remark that both problems admit approximation-preserving reductions from {\sc Vertex Cover}, and hence they too do not admit $(2-\epsilon)$-approximation algorithms up to the Unique Games Conjecture.

We provide  the following results for \fvstfull and \cvdfull.
\begin{itemize}
 \item {\bf Cluster Vertex Deletion.} For any $0<\epsilon<1$, the {\sc Cluster Vertex Deletion} problem admits a $(1+\epsilon)$-approximate $\OO(\frac{1}{\epsilon}\cdot\mathsf{opt})$-vertex kernel.
 
\item {\bf Feedback Vertex Set in Tournaments.} For any $0<\epsilon<1$, the {\sc Feedback Vertex Set in Tournaments} problem admits a $(1+\epsilon)$-approximate $\OO(\frac{1}{\epsilon}\cdot\mathsf{opt})$-vertex kernel.
\end{itemize}

\bigskip
\noindent{\bf Reading Guide.}
First, in Section \ref{sec:prelims}, we present the concept lossy kernelization.  Then, in Section \ref{sec:overview}, we present an overview of our proofs. 
In Section \ref{sec:prelimsAppendix}, we present some basic terminology used throughout the paper. In Section \ref{sec:support}, we present a known result regarding the support of optimum LP solutions to the LP-relaxation of {\sc $d$-Hitting Set}. In Section~\ref{sec:HSElement}, we present our lossy linear-element kernel for {\sc $d$-Hitting Set}. In Section \ref{sec:HSSize}, we present our three lossy kernelization protocols (for {\sc Vertex Cover}, its generalization to {\sc $d$-Hitting Set} with near-linear call size, and a protocol relating the problem to Ruzsa-Szemer\'{e}di graphs). In Section~\ref{sec:lossyImplicit}, we present our $(1+\epsilon)$-approximate linear-vertex kernels for \cvdfull and \fvstfull. 
Lastly, in 
Section~\ref{sec:conclusion}, we conclude with some open problems.  
For easy reference, problem definitions can be found in Appendix~\ref{app:problems}.


\section{Lossy Kernelization: Algorithms and Protocols}\label{sec:prelims}

{\bf Lossy Kernelization Algorithms.} We follow the framework of lossy kernelization presented in \cite{LPRS16}. Here, we deal only with minimization problems where the value of a solution is its size, and where the computation of an arbitrary solution (where no optimization is enforced) is trivial. Thus, for the sake of clarity of presentation, we only formulate the definitions for this context, and remark that the definitions can be extended to the more general setting in the straightforward way (for more information, see \cite{LPRS16}).
To present the definitions, consider a parameterized problem $\Pi$. Given an instance $I$ of $\Pi$ with parameter $k=\kappa(I)$, denote: if $k$ is a structural parameter, then $\pi_I(\mathsf{opt})=\mathsf{opt}$, and otherwise (if $k$ is a bound on the solution size given as part of the input) $\pi_I(\mathsf{opt})=\min\{\mathsf{opt},k+1\}$. Moreover, for any solution $S$ to $I$, denote: if $k$ is a structural parameter, then $\pi_I(S)=|S|$, and otherwise $\pi_I(S)=\min\{|S|,k+1\}$. We remark that when $\pi$ is irrelevant (e.g., when the parameter is structural), we will drop it. A discussion of the motivation behind this definition of $\pi_I$ can be found in~\cite{LPRS16}; here, we only briefly note that it signifies that we ``care'' only for solutions of size at most $k$---all other solutions are considered equally bad, treated as having size $k+1$.

\begin{definition}
Let $\Pi$ be a parameterized minimization problem. Let $\alpha\geq 1$. An {\em $\alpha$-approximate kernelization algorithm} for $\Pi$ consists of two polynomial-time procedures: {\bf reduce} and {\bf lift}. Given an instance $I$ of $\Pi$ with parameter $k=\kappa(I)$, {\bf reduce} outputs another instance $I'$ of $\Pi$  with parameter $k'=\kappa(I')$ such that $|I'|\leq f(k,\alpha)$ and $k'\leq k$. Given $I,I'$ and a solution $S'$ to $I'$, {\bf lift} outputs a solution $S$ to $I$ such that $\displaystyle{\frac{\pi_{I}(S)}{\pi_{I}(\mathsf{opt}(I))}\leq\alpha\frac{\pi_{I'}(S')}{\pi_{I'}(\mathsf{opt}(I'))}}$. 
If $\displaystyle{\frac{\pi_{I}(S)}{\pi_{I}(\mathsf{opt}(I))}\leq\max\{\alpha,\frac{\pi_{I'}(S')}{\pi_{I'}(\mathsf{opt}(I'))}}\}$ holds, then the algorithm is termed~{\em strict}.
\end{definition}
In case $\Pi$ admits an $\alpha$-approximate kernelization algorithm where the output has size $f(k,\alpha)$, or where the output has $g(k,\alpha)$ ``elements'' (e.g., vertices), we say that $\Pi$ admits an $\alpha$-approximate kernel of size $f(k,\alpha)$, or an $\alpha$-approximate $g(k,\alpha)$-element kernel, respectively. When it is clear from context, we simply write $f(k)$ and $g(k)$. When it is guaranteed that $|I'|\leq f(k',\alpha)$ rather than only $|I'|\leq f(k,\alpha)$, then we say that the lossy kernel is {\em \good}.

We only deal with problems that have constant-factor polynomial-time approximation algorithms, and where we may directly work with (the unknown) $\mathsf{opt}$ as the parameter (then, $\pi$ can be dropped). However, working with $k$ (and hence $\pi$) has the effect of artificially altering kernel sizes, but not so if one remembers that $k$ and $\mathsf{opt}$ are different parameterizations. The following lemma clarifies a relation between these two parameterizations. 

\begin{lemma}\label{lem:lossyKerOptToK}
Let $\Pi$ be a minimization problem that, when parameterized by the optimum, admits an $\alpha$-approximate kernelization algorithm $\mathfrak{A}$ of size $f(\mathsf{opt})$ (resp., an $\alpha$-approximate $g(\mathsf{opt})$-element kernel).  Then, when parameterized by $k$, a bound on the solution size that is part of the input, it admits an $\alpha$-approximate kernelization algorithm $\mathfrak{B}$ of size $f(\frac{k+1}{\alpha})$ 
(resp., an $\alpha$-approximate $g(\frac{k+1}{\alpha})$-element kernel).
\end{lemma}

\begin{proof}
We design $\mathfrak{B}$ as follows.  Given an instance $(I,k)$ of $\Pi$, {\bf reduce} of $\mathfrak{B}$ calls {\bf reduce} of $\mathfrak{A}$ on $I$. If the output instance size is at most $f(\frac{k+1}{\alpha})$ 
(resp., the output has at most $g(\frac{k+1}{\alpha})$ elements),
 then it outputs this instance with parameter $k'=k$. Otherwise, it outputs a trivial constant-sized instance. Given $(I,k),(I',k')$ and a solution $S'$ to $(I',k')$, if $I'$ is the output of the {\bf reduce} procedure of $\mathfrak{A}$ on $I$, then {\bf lift} of $\mathfrak{B}$ calls {\bf lift} of $\mathfrak{A}$ on $I,I',S'$ and outputs the result. Otherwise, it outputs a trivial solution to $I$.

The {\bf reduce} and {\bf lift} procedures of  $\mathfrak{B}$ clearly have polynomial time complexities, and the definition of $\mathfrak{B}$ implies the required size (or element) bound on the output of {\bf reduce}. It remains to prove that the approximation ratio is $\alpha$. To this end, consider an input $(I,k),(I',k'),S'$ to {\bf lift} of $\mathfrak{B}$. Let $S$ be its output. We differentiate between two cases.
\begin{itemize}
\item First, suppose that $\mathsf{opt}(I)\geq\frac{k+1}{\alpha}$. Then, $\displaystyle{\frac{\pi_{I}(S)}{\pi_{I}(\mathsf{opt}(I))} \leq \frac{k+1}{\frac{k+1}{\alpha}} = \alpha\leq \alpha\frac{\pi_{I'}(S')}{\pi_{I'}(\mathsf{opt}(I'))}}$ (where the last inequality follows because $|S'|\geq \mathsf{opt}(I')$ and hence $\pi_{I'}(S')\geq\pi_{I'}(\mathsf{opt}(I'))$).

\item Second, suppose that $\mathsf{opt}(I)<\frac{k+1}{\alpha}$. Then, it necessarily holds that $I'$ is the output of the {\bf reduce} procedure of $\mathfrak{A}$ on $I$. Moreover, note that $\mathsf{opt}(I')\leq \mathsf{opt}(I)$ and $k'=k$. So, if $|S'|\geq k'+2$, then $\displaystyle{\frac{\pi_{I}(S)}{\pi_{I}(\mathsf{opt}(I))} \leq \frac{k+1}{\pi_{I}(\mathsf{opt}(I))} =\frac{k'+1}{\mathsf{opt}(I)} \leq \frac{k'+1}{\mathsf{opt}(I')} = \frac{\pi_{I'}(S')}{\pi_{I'}(\mathsf{opt}(I'))}}$. Else, we suppose that $|S'|\leq k'+1$ and hence $\pi_{I'}(S')=|S'|$. Then, 
\[\displaystyle{\frac{\pi_{I}(S)}{\pi_{I}(\mathsf{opt}(I))}\leq \frac{|S|}{\pi_{I}(\mathsf{opt}(I))}=\frac{|S|}{\mathsf{opt}(I)}\leq\alpha\frac{|S'|}{\mathsf{opt}(I')}=\alpha\frac{\pi_{I'}(S')}{\pi_{I'}(\mathsf{opt}(I'))}}.\]
Here, the second inequality follows because the approximation ratio of $\mathfrak{A}$ is $\alpha$.
\end{itemize}
This completes the proof.
\end{proof}

%

Approximate kernelization algorithm often use strict reduction rules, defined as follows.

\begin{definition}
Let $\Pi$ be a parameterized minimization problem. Let $\alpha\geq 1$. An $\alpha$-strict reduction rule for $\Pi$ consists of two polynomial-time procedures: {\bf reduce} and {\bf lift}. Given an instance $I$ of $\Pi$ with parameter $k=\kappa(I)$, {\bf reduce} outputs another instance $I'$ of $\Pi$  with parameter $k'=\kappa(I')\leq k$. Given $I,I'$ and a solution $S'$ to $I'$, {\bf lift} outputs a solution $S$ to $I$ such that $\displaystyle{\frac{\pi_{I}(S)}{\pi_{I}(\mathsf{opt}(I))}\leq\max\{\alpha,\frac{\pi_{I'}(S')}{\pi_{I'}(\mathsf{opt}(I'))}\}}$.
\end{definition}

\begin{proposition}[\cite{LPRS16}]\label{prop:strictRule}
Let $\Pi$ be a parameterized problem. For any $\alpha\geq 1$, an approximate kernelization algorithm for $\Pi$ that consists only of $\alpha$-strict reduction rules has approximation ratio $\alpha$. Furthermore, it is strict.
\end{proposition}

\medskip
\noindent{\bf Lossy Kernelization Protocols.} We extend the notion of lossy kernelization algorithms to lossy kernelization protocols as follows.

\begin{definition}[{\bf Lossy Kernelization Protocol}]
Let $\Pi$ be a parameterized minimization problem with parameter $k$. Let $\alpha\geq 1$. An {\em $\alpha$-approximate kernelization protocol} of {\em call size $f(k,\alpha)$} and $g(k,\alpha)$ {\em rounds} for $\Pi$ is defined as follows.
First, the protocol assumes to have access to an oracle $\mathfrak{O}$ that, given an instance $I'$ of $\Pi$ of size at most $f(k,\alpha)$, returns a solution $S'$ to $I'$ such that $\displaystyle{\pi_{I'}(S')\leq\beta\pi_{I'}(\mathsf{opt}(I'))}$ for minimization and $\displaystyle{\pi_{I'}(S')\geq\frac{1}{\beta}\pi_{I'}(\mathsf{opt}(I'))}$ for maximization, for some fixed $\beta>0$. Second, for the same fixed $\beta>0$, given an instance $I$ of $\Pi$, the protocol may perform $g(k,\alpha)$ calls to $\mathfrak{O}$ and other operations in polynomial time, and then output a solution $S$ to $I$ such that $\displaystyle{\frac{\pi_{I}(S)}{\pi_{I}(\mathsf{opt}(I))}\leq\alpha\beta}$.

The {\em volume} (or {\em size}) of the protocol is $f(k,\alpha)g(k,\alpha)$. In case $g(k,\alpha)=g(\alpha)$ (i.e., $g$ depends only on $\alpha$), the protocol is called {\em pure}.
\end{definition}
Notice that an $\alpha$-approximate kernelization algorithm is the special case of an $\alpha$-approximate kernelization protocol when the number of rounds is $1$.

Practically, we think that (lossy) kernelization protocols can often be as useful as standard (lossy) kernels, and, in some cases, more useful. Like standard (lossy) kernels, they reduce the total size of what we need to solve, only that now what we need to solve is split into several instances, to be solved one after another. On the one hand, this relaxation seems to, in most cases, not be restrictive (as what we really care about is the total size of what we need to solve). On the other hand, it might be helpful if by using this relaxation one can achieve better bounds than what is known (or, even, what is possible) on the sizes of the reduced instances, or to simplify the algorithm. For example, for the case of {\sc $d$-Hitting Set}, we do not know how to beat $\OO(k^d)$ using a lossy kernel rather than a protocol.


\section{Overview of Our Proof Ideas}\label{sec:overview}

In this section, we present a high-level overview of our proof ideas. For standard terminology not defined here or earlier, we refer the reader to Section \ref{sec:prelimsAppendix}.

\subsection{Linear-Element Lossy Kernel for {\sc $d$-Hitting Set}} 
We make use of a known result about the natural LP relaxation of  {\sc $d$-Hitting Set}: the support of any optimal LP solution to the LP-relaxation of {\sc $d$-Hitting Set} is of size at most $d\cdot\mathsf{frac}$ where $\mathsf{frac}$ is the optimum (minimum value) of the LP~\cite{furedi1988matchings}. For the sake of completeness, we provide a proof. We then provide a lossy reduction rule that computes an optimal LP solution, and deletes all vertices assigned values at least $\frac{1}{d-1}$. Having applied this rule exhaustively, we arrive at an instance having an optimal LP solution that assigns only values strictly smaller than $\frac{1}{d-1}$. Then, it can be shown that all hitting sets are contained within the support of this LP solution. In turn, in light of the aforementioned known result, this yields an approximate $d\cdot\mathsf{frac}$-element and $(d\mathsf{frac})^d$-set kernel that is \good. 

The analysis that the approximation factor is $d-\frac{d-1}{d}$ is slightly more involved, and is based on case distinction. In case the number of vertices deleted is ``small enough'', the cost of adding them is ``small enough'' as well. In the more difficult case where the number of vertices deleted is ``large'', by making use of the already established bound on the output size as well as the drop in the fractional optimum, we are able to show that, in fact, we return a solution of  approximation factor $d-\frac{d-1}{d}$ irrespective of the approximation ratio of the solution we are given. More generally, the definition of ``small enough'' and ``large'' gives rise to a trade-off that is critical for our kernelization protocol for {\sc $d$-Hitting Set}, which in particular yields that we can either obtain a {\em negligible additive error} or directly a solution of the desired (which is some fixed constant better than $d$ but worse than $d-\frac{d-1}{d}$) approximation ratio. Specifically, this means that it is ``safe'' to compose our element kernel as part of other kernelization algorithms or protocols.

\begin{figure}
  \begin{center}\fbox{\includegraphics[scale=0.5]{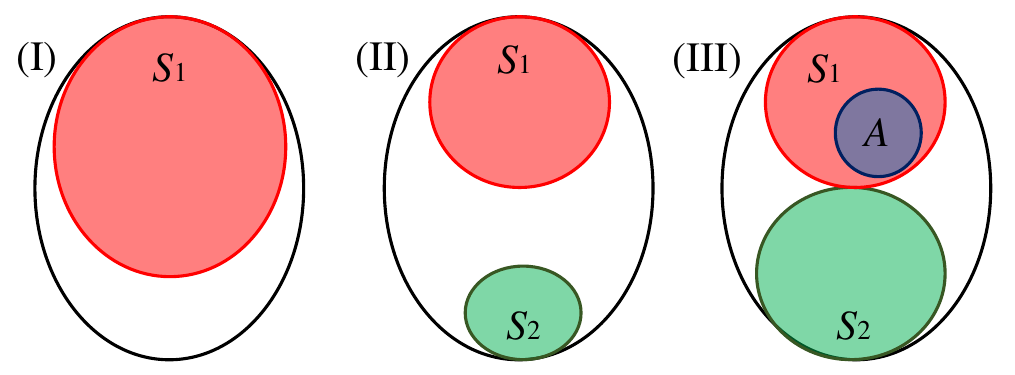}}
  \caption{The three cases encountered by our $2$-call lossy kernelization protocol for {\sc Vertex Cover}: (I) $|S_1|$ is large, and we return $V(G)$; (II) $|S_1|$ is small and $|S_2|$ is small, and we return $S_1\cup S_2$; (III) $|S_1|$ is small and $|S_2|$ is large, and we return $(V(G)\setminus S_1)\cup A$.}
  \label{fig:vc}
  \end{center}
\end{figure}

\subsection{2-Round $\OO(\mathsf{frac}^{1.5})$-Call Size Lossy Kernelization Protocol for {\sc Vertex Cover}}  Towards the presentation of our near-linear call size lossy kernelization protocol for {\sc $d$-Hitting Set}, we abstract some of the ideas using a simpler 2-round $\OO(\mathsf{frac}^{1.5})$-call size $1.721$-approximate kernelization protocol for {\sc Vertex Cover} (where $\mathsf{frac}\leq\mathsf{opt}\leq k$ is the optimum of the natural LP relaxation of {\sc Vertex Cover}). First, we apply an (exact) kernelization algorithm to have a graph $G$ on at most $2\mathsf{frac}$ vertices. The purpose of having only $2\mathsf{frac}$ vertices is twofold. First, it means that to obtain a ``good enough'' approximate solution, it suffices that we do not pick a ``large enough'' (linear fraction) of vertices of $G$ to our solution. Second, it is required for a probability bound derived using union bound over vertex subsets to hold. Then, roughly speaking, the utility of the first oracle call is mainly, indirectly, to uncover a ``large'' (linear in $n=|V(G)|$) induced subgraph of $G$ that is ``sparse'', and hence can be sent to the second oracle call to be solved optimally. 

More precisely, after applying the initial kernelization, we begin by sampling roughly $\mathsf{frac}^{1.5}$ edges from $G$. Then, we call the oracle on the sampled graph to obtain a solution $S_1$ to it (but not to $G$). In case that solution $S_1$ is ``large'' compared to the size of the vertex set of $G$ (that is, sufficiently larger than $n/2\leq \mathsf{frac}$), we can just return the entire vertex set of $G$ (see Fig.~\ref{fig:vc}). Else, we know that the subgraph of the sampled graph that is induced by $V(G)\setminus S_1$ is edgeless. In addition, we can show (due to the initial kernelization) that with high probability, every set of edges of size (roughly) at least $\mathsf{frac}^{1.5}$ that is the edge set of some induced subgraph of $G$ has been hit by our edge sample. Together, this implies that the subgraph of $G$ induced by $V(G)\setminus S_1$ has at most $\mathsf{frac}^{1.5}$ edges, and hence can be solved optimally by a second oracle call. Then, because we know that this subgraph is large compared to $G$ (else $S_1$ is large), if the oracle returned a ``small'' solution $S_2$ to it, we may just take this solution together with $S_1$ (which will form a vertex cover), and yet not choose sufficiently many vertices so that this will be good enough in terms of the approximation ratio achieved. Else, also because we know that this subgraph is large compared to $G$, if the second oracle returned a ``large'' solution $S_2$, then we know that every optimal solution must take many vertices from this subgraph, and hence, to compensate for this, the optimum of $G[S_1]$ must be ``very small''. So, we compute a $2$-approximate solution $A$ to $G[S_1]$, which we know should not be ``too large'', and output the union of $A$ and $V(G)\setminus S_1$ (which yields a vertex cover).

\subsection{Near-Linear Volume and Pure Lossy Kernelization Protocol for {\sc $d$-Hitting Set}} For any fixed $\epsilon>0$, we present a pure $d(1-h(d,\epsilon))$-approximate (randomized) kernelization protocol for {\sc $d$-Hitting Set} with call size $\OO((\mathsf{frac})^{1+\epsilon})$ where $h(d,\epsilon)$ is a fixed positive constant that depends only on $d,\epsilon$. On a high-level, the idea of our more general lossy kernelization protocol is to compute a nested family of solutions based on the approach described above for {\sc Vertex Cover} (see Fig.~\ref{fig:hs}). Intuitively, as we now can sample only few sets (that is, $\mathsf{frac}^{1+\epsilon}$), when we compute a solution that hits them using an oracle call, the number of sets it misses can still be huge, and hence we will need to iteratively use the oracle (a constant number of times) until we reach   a subuniverse such that we can optimally solve the subinstance induced by it by a single oracle call. Below, we give a more elaborate overview.

First, we apply our linear-element lossy kernel to have an instance $I_0=(U_0,{\cal T}_0)$ where the universe $U_0$ consists of at most $d\mathsf{frac}$ elements. Here, the error of this application is not multiplied by the error attained next, but will only yield (as mentioned earlier) a negligible additive error (or directly a solution of the desired approximation ratio). The purpose of having only $d\mathsf{frac}$ elements is twofold, similarly as it is in the protocol described earlier for {\sc Vertex Cover}. Afterwards, we begin by sampling a family ${\cal F}_1$ of roughly $\mathsf{frac}^{1+\epsilon}$ sets from ${\cal T}_0$. Then, we call the oracle on the sampled family ${\cal F}_1$ to obtain a solution $S_1$ to it. In case that solution $S_1$ is ``large'' (sufficiently larger than $|U_0|/d\leq \mathsf{frac}$), we can just return $U_0$. Else, we know that the family of sets corresponding to the subinstance $I_1$ induced by $U_1=U_0\setminus S_1$---that is, the family of all sets in ${\cal T}_0$ contained in $U_1$, which we denote by ${\cal T}_1$---was missed by our set sample. In addition, we can show (due to the initial kernelization) that with high probability, every family of sets of size (roughly) at least $\mathsf{frac}^{d-\epsilon}$ that corresponds to a subinstance induced by a subset of $U_0$ has been hit by our set sample. Together, this implies that ${\cal T}_1$ has at most $\mathsf{frac}^{d-\epsilon}$ (rather than $\mathsf{frac}^{d}$) sets. Hence, in some sense, we have made progress towards the discovery of a sparse subinstance that we can optimally solve.

\begin{figure}
  \begin{center}\fbox{\includegraphics[scale=0.5]{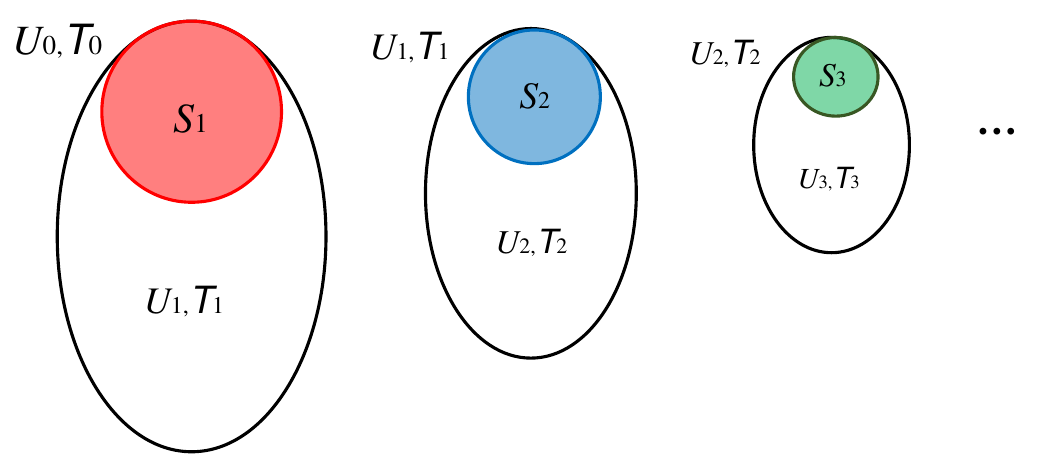}}
  \caption{The nested solutions computed by oracle calls in our lossy kernelization protocol for {\sc $d$-Hitting Set}. Each $S_i$ is a solution to a subinstance $(U_{i-1},{\cal F}_{i-1})$ sampled from $(U_{i-1},{\cal T}_{i-1})$.}
  \label{fig:hs}
  \end{center}
\end{figure}

Due to important differences, let us describe also the second iteration---among at most $\frac{1}{\epsilon}(d-1)$ iterations performed in total---before skipping to the (last) one where we have a subinstance that we can optimally solve by an oracle call. The last iteration may not even be reached, if we find a ``good enough'' solution earlier. We remark that it is critical to stop and return a solution as soon as we find a ``large enough'' one by an oracle call\footnote{The solution we return is not the one given by the oracle call, but its union with another solution, as will be clarified immediately, or just $U_0$ in case of the first iteration describe above.} as for our arguments to work, we need to always deal with subinstances whose universe is large (a linear fraction of $|U_0|$), and these are attained by removing oracle solutions we got along the way. We begin the second iteration by sampling a family ${\cal F}_2$ of roughly $\mathsf{frac}^{1+\epsilon}$ sets from ${\cal T}_1$. Then, we call the oracle on the sampled family ${\cal F}_2$ to obtain a solution $S_2$ to it. On the one hand, in case that solution $S_2$ is ``large'' (sufficiently larger than $|U_1|/d$), we cannot just return $U_0$ as in the first iteration, as now it may not be true that the optimum of $I_0$ is large compared to $|U_0|$.  Still, it is true that the optimum of $I_1$ is large compared to $|U_1|$. So,  every optimal solution (to $I_0$) must take many elements from $U_1\setminus S_2$, and hence, to compensate for this, the optimum of the subinstance induced by $S_1$ must be ``very small''. So, we compute a $d$-approximate solution to this subinstance, which we know should not be ``too large'' , and output the union of it and $U_1$ (which yields a hitting set). On the other hand, in case $S_2$ is ``small'', we proceed as follows. We observe that the family of sets corresponding to the subinstance $I_2$ induced by $U_2=U_1\setminus S_2$, whose family of sets we denote by ${\cal T}_2$, was missed by our set sample. In addition, we can show (due to the initial kernelization) that with high probability, every family of sets of size (roughly) at least $\mathsf{frac}^{d-2\epsilon}$ that corresponds to a subinstance induced by a subset of $U_1$ has been hit by our set sample. Together, this implies that ${\cal T}_2$ has at most $\mathsf{frac}^{d-2\epsilon}$ (rather than just $\mathsf{frac}^{d-\epsilon}$ as in the first iteration) sets. Hence, in some sense, we have made further progress towards the discovery of a sparse subinstnace that we can optimally solve.

Finally, we arrive at a subinstance $I'$ induced by a subuniverse $U'\subseteq U_0$ that is of size linear in $U_0$ (else we should have returned a solution earlier) and where the family of sets, ${\cal F}'$, is of size at most $\mathsf{frac}^{1+\epsilon}$. Then, we call the oracle on $I'$ to obtain a solution $S'$ to it. On the one hand, in case that solution $S'$ is ``large'' (sufficiently larger than $|U'|/d$), we compute a $d$-approximate solution to the subinstance induced by $U_0\setminus U'$ (which is the union of all solutions returned by oracle calls except the last one), and output the union of it and $U'$. Otherwise, we output $(U_0\setminus U')\cup S'$, which is ``good enough'' because $U'$ is sufficiently large while $S'$ is sufficiently small compared to it, it does not contain a ``large enough'' number of elements from $U_0$.

\subsection{Outlook: Relation to Ruzsa-Szemer\'{e}di Graphs}   A graph $G$ is an {\em $(r,t)$-Ruzsa-Szemer\'{e}di graph} if its edge set can be partitioned into $t$ edge-disjoint induced matchings, each of size $r$. These graphs were introduced in 1978~\cite{ruzsa1978triple},  and have been extensively studied since then. When $r$ is a function of $n$, let $\gamma(r)$ denote the maximum $t$ (which is a function of $n$) such that there exists an $(r,t)$-Ruzsa-Szemer\'{e}di graph. 
 In~\cite{fox2017graphs}, the authors considered the case where $r=cn$. They showed that when $c=\frac{1}{4}$, $\gamma(r)\in\Theta(\log n)$, and when $\frac{1}{5}\leq c\leq\frac{1}{4}$, $t\in\OO(\frac{n}{\log n})$. It is an open problem whether whenever $c$ is a fixed constant, $t\in\OO(n^{1-\epsilon})$. For any fixed constant $0<c<\frac{1}{4}$, we present a $(1+4c)$-approximate (randomized) kernelization protocol for {\sc Vertex Cover} with $t+1$ rounds and call size $\OO(t(\mathsf{frac})^{1.5})$. Clearly, this result makes sense only when $t\in o(\sqrt{n})$, preferably $t\in\OO(n^{\frac{1}{2}-\lambda})$ for $\lambda$ as close to $1/2$ as possible, because the volume is $\OO(\mathsf{opt}^{2-\lambda})$. If $t$ is ``sufficiently small'' (depending on the desired number of rounds) whenever $c$ is a fixed constant (specifically, substitute $c=\frac{\epsilon}{4}$), this yields a $(1+\epsilon)$-approximate kernelization protocol. 

We observe that, for a graph $G$, $r=r(n),t=t(n)\in\mathbb{N}$ and $U_1,U_2,\ldots,U_t\subseteq V(G)$ such that for all $i\in \{1,2,\ldots,t\}$, $G[U_i]$ has a matching $M_i$ of size at least $r$, and for all distinct $i,j\in \{1,2,\ldots,t\}$, $E(G[U_i])\cap E(G[U_j])=\emptyset$, we have that $G$ is a supergraph of an $(r,t)$-Ruzsa-Szemer\'{e}di graph. Having this observation in mind, we devise our protocol as follows. After applying an exact $2\mathsf{frac}$-vertex kernel, we initialize $E'=\emptyset$, and we perform $t+1$ iterations of the following procedure. We sample a set of roughly $\mathsf{frac}^{1.5}$ edges from $G$, and call the oracle on the subgraph of $G$ whose edge set is the set of samples edges union $E'$ to obtain a solution $S$ to it (but not to $G$), and compute a maximal matching $M$ in $G-S$. If $|M|$ is smaller than $cn\leq 2c\mathsf{frac}$, then we return the union of the set of vertices incident to edges in $M$ (which is a solution to $G-S$) and $S$. Else, similarly to the first protocol we described for {\sc Vertex Cover}, we can show that with high probability, $G-S$ has (roughly) at most $k^{1.5}$ edges, and we add this set of edges to $E'$. The crux of the proof is in the argument that, at the latest, at the $(t+1)$-st iteration the computed matching will be of size smaller than $cn\leq 2c\mathsf{frac}$, as otherwise we can use the matchings we found, together with the vertex sets (of the form $G-S$) we found them in, to construct an $(r,t+1)$-Ruzsa-Szemer\'{e}di graph based on the aforementioned observation, which contradicts the choice of $t$.

\subsection{$(1+\epsilon)$-Approximate $\OO(\frac{1}{\epsilon}\cdot\mathsf{opt})$-Vertex Kernel for Implicit {\sc $3$-Hitting Set} Problems} Both of our lossy kernels share a common scheme, which might be useful to derive $(1+\epsilon)$-approximate linear-vertex kernels for other implicit hitting and packing problems as well. Essentially, they both consist of two rules (although in the presentation, they are merged for simplicity). To present them, we remind that a module (in a graph) is a set of vertices having the same neighborhood relations with all vertices outside the set. Now, our first rule reveals some modules  in the graph, and our second rule shrinks their size. The first rule in both of our lossy kernels is essentially the same.

Now, we elaborate on the first rule. We start by computing an optimal solution $\alpha$ to the LP-relaxation of the corresponding {\sc $3$-Hitting Set} problem. Notice that $\mathsf{support}(\alpha)$ is a solution, and its size is at most 3$\mathsf{frac}$ (in fact, we show that it is at most $3\mathsf{frac}-2|\alpha^{-1}(1)|$). Then, the first rule is as follows. At the beginning, no vertex is marked. Afterwards, one-by-one, for each vertex $v$ assigned $1$ by $\alpha$ (i.e., which belongs to $\alpha^{-1}(1)$), we construct a graph whose vertex set is the set of yet unmarked vertices in $V(G)\setminus\mathsf{support}(\alpha)$ and where there is an edge between every two vertices that create an obstruction together with $v$ (that is, an induced $P_3$ in {\sc Cluster Vertex Deletion} and a triangle in {\sc Feedback Vertex Set in Tournaments}). We compute a maximal matching in this graph, and decrease its size to $\frac{1}{\epsilon}$ if it is larger (in which case, it is no longer maximal). The vertices incident to the edges in the matching are then considered marked. 
We prove that among the vertices in $\alpha^{-1}(1)$ whose matching size was decreased, whose set is denoted by $D$, any solution can only exclude an $\epsilon$ fraction of its size among the vertices in $D$, and hence it is ``safe'' (in a lossy sense) to delete $D$. Let $M$ be the set of all marked vertices. Then, we show that $(\mathsf{support}(\alpha)\cup M)\setminus\{v\}$, for any $v\in\mathsf{support}(\alpha)$ (including those not in $\alpha^{-1}(1)$), is also a solution. 

For {\sc  Cluster Vertex Deletion}, we prove that the outcome of the first rule means that the vertex set of every clique in $G-(\mathsf{support}(\alpha)\cup M)$ is a module in $G-D$, and that for every vertex in $\mathsf{support}(\alpha)$, the set of its neighbors in $V(G-(\mathsf{support}(\alpha)\cup M))$  is the vertex set of exactly one of these cliques. So, for {\sc  Cluster Vertex Deletion}, this gives rise to the following second reduction rule (which is, in fact, exact) to decrease the size of module. For every clique among the aforementioned cliques whose size is larger than that of its neighborhood, we arbitrarily remove some of its vertices so that its size will be equal to the size of its neighborhood. This rule is safe since if at least one of the vertices in such a clique is deleted by a solution, then because it is a module, either that deletion is irrelevant or the entire clique is deleted, and in the second case we might just as well delete its neighborhood instead. Because the neighborhoods of the cliques are pairwise-disjoint  (since for every vertex in $\mathsf{support}(\alpha)$, the set of its neighbors in $V(G-(\mathsf{support}(\alpha)\cup M))$  is the vertex set of exactly one of the cliques), this means that now their total size is at most $(\mathsf{support}(\alpha)\setminus D)\cup M$, and hence we arrive at the desired kernel.

For {\sc Feedback Vertex Set in Tournaments}, we consider the unique (because $G$ is a tournament) topological ordering of the vertices in $G-\mathsf{support}(\alpha)$, so that all arcs are ``forward'' arcs. We prove that the outcome of the first rule means that each vertex $v\in\mathsf{support}(\alpha)$ has a unique position within this ordering when restricted to $G-(\mathsf{support}(\alpha)\cup M)$, so that still all arcs (that is, including those incident to $v$) are forward arcs in $G-(\mathsf{support}(\alpha)\cup M)\cup\{v\}$. (Further, the vertex set of each subtournament induced by the vertices ``between'' any two marked vertices in $G-\mathsf{support}(\alpha)$ is a module in $G-D$.) We are thus able to characterize  all triangles in $G-D$ as follows: each either consists of three vertices in $(\mathsf{support}(\alpha)\setminus D)\cup M$, or it consists of a vertex $v\in\mathsf{support}(\alpha)\setminus D$, a vertex $u\in(\mathsf{support}(\alpha)\setminus D)\cup M$ and a vertex $w\in V(G)\setminus (\mathsf{support}(\alpha)\cup M)$ with a backward arc between $v$ and $u$ and where $w$ is ``in-between'' the positions of $v$ and $u$. This gives rise to a reduction rule for module shrinkage whose presentation and analysis are more technical than that of {\sc  Cluster Vertex Deletion}  (in particular, unlike the second rule of {\sc  Cluster Vertex Deletion}, the second rule of {\sc Feedback Vertex Set in Tournaments} is lossy) and of the first rule, and hence we defer them to the appropriate Section \ref{sec:fvst}.


\section{Preliminaries}\label{sec:prelimsAppendix}

\subsection{General Notation}
The {\em support} of a function $f: A\rightarrow\mathbb{R}$ is $\{a\in A: f(a)\neq 0\}$, denoted by $\mathsf{support}(f)$.

Given an instance $I$ of some optimization problem $\Pi$, we denote by $\mathsf{opt}(I)$ the optimum (value of an optimal solution, if one exists) of $I$. When $I$ is clear from context, we simple write $\mathsf{opt}$.

To bound the approximation ratios of our algorithms, we will use the following fact.

\begin{proposition}[Folklore, see, e.g., \cite{LPRS16}]\label{prop:numbers}
For any positive reals  $x,y,p$ and $q$, $\min\left(\frac{x}{p},\frac{y}{q}\right)\leq \frac{x+y}{p+q} \leq \max\left(\frac{x}{p},\frac{y}{q}\right)$.
\end{proposition} 

We now present a well-known Chernoff bound, to be used in the analysis of our (randomized) lossy kernelization protocols.

\begin{proposition}\label{prop:chernoff}
Let $X_1, ..., X_n$ be independent random variables over $\{0,1\}$.  Let $X$ denote their sum and let $\mu = \mathsf{E}[X]$ denote the expected value of $X$. Then, for any $0\leq\delta\leq 1$,
\[\mathsf{Prob}[X\geq (1+\delta)\mu] \leq e^{-\frac{\delta^2\mu}{3}}.\]
\end{proposition}

\subsection{Graph Notation}

Given a graph $G$, let $V(G)$ and $E(G)$ denote its vertex set and edge (or arc) set, respectively. When clear from context, $n=|V(G)|$ and $m=|E(G)|$. Given a vertex $v\in V(G)$, let $N_G(v)$ denote the set of neighbors of $v$ in $G$, and given a subset $U\subseteq V(G)$, let $N_G(U)$ denote the open neighborhood of $U$ in $G$. Given a subset $U\subseteq V(G)$, let $G[U]$ denote the subgraph of $G$ induced by $U$, that is, the graph on vertex set $U$ and edge set $\{\{u,v\}\in E(G): u,v\in U\}$. Moreover, given a subgraph $G'$ of $G$ (possibly $G'=G$) and a subset $U\subseteq V(G)$ (possibly $U\setminus V(G')\neq\emptyset$) , let $G'-U$ denote the graph on vertex set $V(G')\setminus U$ and edge set $\{\{u,v\}\in E(G'): u,v\notin U\}$. A {\em module} in $G$ is a subset $U\subseteq V(G)$ such that for every vertex $v\in V(G)\setminus U$ either $U\subseteq N_G(v)$ or $U\cap N_G(v)=\emptyset$. Given a subset $W\subseteq E(G)$, let $G-W$ denote the graph on vertex set $V(G)$ and edge set $E(G)\setminus W$. An {\em induced $P_3$} in $G$ is a path on three vertices in $G$ whose endpoints are not adjacent in $G$. A {\em cluster graph} is a graph in which every connected component is a clique. An {\em acyclic} digraph is a digraph that contains no directed cycles. A {\em tournament} is a digraph where for every two vertices $u,v$, exactly one among the arcs $(u,v)$ and $(v,u)$ belongs to the digraph.


\begin{definition}\label{def:ruzsa}
 A graph $G$ is an {\em $(r,t)$-Ruzsa-Szemer\'{e}di graph} if its edge set can be partitioned into $t$ edge-disjoint induced matchings, each of size $r$.
\end{definition}

These graphs were introduced in 1978~\cite{ruzsa1978triple},  and have been extensively studied since then. In \cite{fox2017graphs}, the authors considered the case where $r=cn$. They showed that when $c=\frac{1}{4}$, the maximum $t$, which we denote by $\gamma(r)$, is $\Theta(\log n)$, and when $\frac{1}{5}\leq c\leq\frac{1}{4}$, $t=\OO(\frac{n}{\log n})$. It is an open problem whether when $c$ is a fixed constant, $t=\OO(n^{1-\epsilon})$.

\subsection{Linear Programming}
A canonical form of a linear program (LP) is  $[\max \sum_{i=1}^nc_ix_i$ s.t.~$\forall j=1,\ldots,m: \sum_{i=1}^na_{ji}x_i\leq b_j; \forall i=1,\ldots,n: x_i\geq 0]$, or $[\min \sum_{j=1}^mb_jy_j$ s.t.~$\forall i=1,\ldots,n: \sum_{j=1}^ma_{ij}y_j\geq c_i; \forall j=1,\ldots,m: y_j\geq 0]$. Here, the $x_i$'s ($y_j$'s) are variables.  Moreover, two programs of the aforementioned forms that refer to the same set of coefficients $\{c_i\}|_{i=1}^n,\{a_{ij}: i\in\{1,\ldots,n\},j\in\{1,\ldots,m\}\},\{b_j\}|_{j=1}^m$ are {\em dual} of each other. A {\em solution} to an LP is an assignment of real values to its variables so that all constraints are satisfied. Further, a solution is {\em optimal} is it also optimizes (maximizes or minimizes) the value of the objective function. The optimum (value of an optimal solution, if one exists) of an LP $I$ (or which is associated with some entity $I$, where no confusion can arise) is denoted by $\mathsf{frac}(I)$. When $I$ is clear from context, we simple write $\mathsf{frac}$.

\begin{proposition}[\cite{matousek2007understanding}]\label{prop:solveLP}
Any LP (with rational coefficients) that admits a solution, admits an optimal solution that assigns only rational values. Furthermore, such an optimal solution an be computed in polynomial time.
\end{proposition}

We will need a well-known proposition relating optimal solutions to LPs and their duals, known as strong duality and complementary slackness: 

\begin{proposition}[\cite{matousek2007understanding}]\label{prop:slackness}
Let {\bf (P)} $[\max \sum_{i=1}^nc_ix_i$ s.t.~$\forall j=1,\ldots,m: \sum_{i=1}^na_{ji}x_i\leq b_j; \forall i=1,\ldots,n: x_i\geq 0]$ be a primal LP; {\bf (D)} $[\min \sum_{j=1}^mb_jy_j$ s.t.~$\forall i=1,\ldots,n: \sum_{j=1}^ma_{ij}y_j\geq c_i; \forall j=1,\ldots,m: y_j\geq 0]$ be the dual LP. Let $\alpha$ and $\beta$ be solutions to {\bf (P)} and {\bf (D)}, respectively. Then, $\alpha$ and $\beta$ are both optimal  if and only if $\sum_{i=1}^nc_i\alpha(x_i)=\sum_{j=1}^mb_j\beta(y_j)$ [strong duality]. Moreover, $\alpha$ and $\beta$ are both optimal  if and only if [complementary slackness]:
\begin{itemize}
\item For $i=1,\ldots,n$: $\alpha(x_i)>0$ if and only if $\sum_{j=1}^ma_{ij}\beta(y_j)= c_i$.
\item For $j=1,\ldots,m$: $\beta(y_j)>0$ if and only if $\sum_{i=1}^na_{ji}\alpha(x_i)= b_j$.
\end{itemize}
\end{proposition}


\section{The Support Size of Any Optimal Solution to the LP of {\sc $d$-Hitting Set}}\label{sec:support}

In this section, we present a tight bound on the support size of any optimal solution to the classic LP of the {\sc $d$-Hitting Set} problem, defined as follows. 

\begin{definition}\label{def:HSLP}
Let $(U,{\cal F})$ be an instance of {\sc $d$-Hitting Set}. Then, the classic LP that corresponds to $(U,{\cal F})$ is defined as follows: $[\min \sum_{u\in U} y_u$ s.t.~$\forall S\in{\cal F}: \sum_{u\in S}y_u\geq 1; \forall u\in U: y_u\geq 0]$.
\end{definition}

We will re-name $y$ by $x$ when it is more convenient (in Section \ref{sec:HSElement}) and no confusion arises.

We present the following theorem, which has been originally proved in \cite{furedi1988matchings}. For the sake of completeness, we present a short proof here.

\begin{theorem}[\cite{furedi1988matchings}]\label{thm:support}
Let $I=(U,{\cal F})$ be an instance of {\sc $d$-Hitting Set}. Let $\beta$ be an optimal solution to its classic LP. Then, $|\mathsf{support}(\beta)|\leq d\cdot \mathsf{frac}(I)$. 
In particular, $|\mathsf{support}(\beta)|\leq d\cdot \mathsf{opt}(I)$. 
\end{theorem}

\begin{proof} Let us denote the classic LP that corresponds to $(U,{\cal F})$ by {\bf (D)}. We note that the dual LP of {\bf (D)}, which we denote by {\bf (P)}, is defined as follows: $[\max \sum_{S\in{\cal F}} x_S$ s.t.~$\forall u\in U: \sum_{S\in {\cal F}: u\in S}x_S\leq 1; \forall S\in{\cal F}: x_S\geq 0]$. Let $\alpha$ be an optimal solution to {\bf (P)}. Then,
\[\begin{array}{lll}
\smallskip
\mathsf{frac}(I) & = \displaystyle{\sum_{u\in U} \beta(y_u)} & [\beta \mbox{ is optimal}]\\
\smallskip
& = \displaystyle{\sum_{S\in{\cal F}} \alpha(x_S)} & [\mbox{strong duality}]\\
\smallskip
& = \displaystyle{\frac{1}{d}\cdot \sum_{S\in{\cal F}} (d\cdot \alpha(x_S))} & \\
\smallskip
& \geq \displaystyle{\frac{1}{d}\cdot \sum_{S\in{\cal F}}\sum_{u\in S}\alpha(x_S)} & [\forall S\in {\cal F}: |S|\leq d]\\
\smallskip
& = \displaystyle{\frac{1}{d}\cdot \sum_{u\in U}\sum_{S\in{\cal F}: u\in S}\alpha(x_S)} &\\
\smallskip
& \geq \displaystyle{\frac{1}{d}\cdot \sum_{u\in \mathsf{support}(\beta)}\sum_{S\in{\cal F}: u\in S}\alpha(x_S)} &\\
\smallskip
& = \displaystyle{\frac{1}{d}\cdot \sum_{u\in \mathsf{support}(\beta)}1} & [\mbox{complementary slackness}]\\
& = \displaystyle{\frac{1}{d}\cdot |\mathsf{support}(\beta)|} & \\
\end{array}\]
We conclude that $|\mathsf{support}(\beta)|\leq d\cdot \mathsf{opt}(I)$. Because $\mathsf{frac}(I)\leq \mathsf{opt}(I)$, the proof is complete. 
\end{proof}

Observe that the bound in Theorem \ref{thm:support} is tight, that is, it is satisfied with equality for infinitely many instances of {\sc $d$-Hitting Set}. To see this, for any $n\in\mathbb{N}$ that is a multiple of $d$, consider an instance $I=(U,{\cal F})$ where $|U|=n$ and ${\cal F}$ is a partition of $U$ into parts of equal size $d$ (so, $|{\cal F}|=n/d$). Then, the optimum of the corresponding classic LP is easily seen to be $n/d$, and it can be attained by an assignment that assigns $1/d$ to each variable, and thus has support size $n=d\cdot\mathsf{opt}(I)$.


\section{A $(d-\frac{d-1}{d})$-Approximate Linear-Element Kernel for {\sc $d$-Hitting Set}}\label{sec:HSElement}

We first present the following reduction rule that is the basis of our kernelization algorithm.

\begin{definition}\label{def:HSRule}
The {\em {\sc $d$-Hitting Set} element reduction rule} is defined as follows:
\begin{itemize}
\item {\bf reduce:} Let $I=(U,{\cal F})$ be an instance of {\sc $d$-Hitting Set}. Use the algorithm in Proposition \ref{prop:solveLP} to compute an optimal solution $\alpha$ to the classic LP corresponding to it (Definition \ref{def:HSLP}). Let $H=\{u\in U: \alpha(u)\geq\frac{1}{d-1}\}$.
Output $I'=(U',{\cal F}')$ where ${\cal F}'=\{S\in{\cal F}: S\cap H=\emptyset\}$ and $U'=\bigcup{\cal F}'$.
\item {\bf lift:} Given $I,I'$ and a solution $S'$ to $I'$, output $S=S'\cup H$.
\end{itemize}
\end{definition}

Essentially, our approximate kernelization algorithm will consist of exhaustive (i.e., as long as $|H|\geq 1$) application of the {\sc $d$-Hitting Set} element reduction rule. Unfortunately, the $d$-Hitting Set rule is {\em not} $(1-\frac{d-1}{d})$-strict, and hence, unlike other lossy kernelization algorithms that consist of repetitive applications of one or more reduction rules, we cannot make direct use of Proposition \ref{prop:strictRule}. So, we present the algorithm explicitly in order to ease its analysis. 

\begin{definition}\label{def:HSKernelization}
The {\em {\sc $d$-Hitting Set} element kernelization algorithm} is defined as follows:
\begin{itemize}
\item {\bf reduce:} Let $I=(U,{\cal F})$ be an instance of {\sc $d$-Hitting Set}. Let $i=1,{\cal F}_1={\cal F}$ and $U_1=\bigcup{\cal F}_1$. As long as a break command is not reached:
	\begin{enumerate}
	\item Use the algorithm in Proposition \ref{prop:solveLP} to compute an optimal solution $\alpha_i$ to the classic LP corresponding to $I_i=(U_i,{\cal F}_i)$ (Definition \ref{def:HSLP}).
	\item Let $H_i=\{u\in U: \alpha(u)\geq\frac{1}{d-1}\}$. If $H_i=\emptyset$, then break the loop.
	\item Increase $i$ by $1$, and let ${\cal F}_i=\{S\in{\cal F}_{i-1}: S\cap H_{i-1}=\emptyset\}$ and $U_i=\bigcup{\cal F}_i$.
	\end{enumerate}
Let $H^\star=\bigcup_{j=1}^{i-1}H_j$. Output $I'=(U',{\cal F}')$ where ${\cal F}'=\{S\in{\cal F}: S\cap H^\star=\emptyset\}$ (which equals ${\cal F}_i$)  and $U'=\bigcup{\cal F}'$ (which equals $U_i$).

\item {\bf lift:} Given $I,I'$ and a solution $S'$ to $I'$, output $S=S'\cup H^\star$.
\end{itemize}
\end{definition}

In order to bound the output size of our kernelization algorithm, we will make use of the following lemma, whose proof is based on Theorem \ref{thm:support}.

\begin{lemma}\label{lem:dHSFinalElemSize}
Let $I=(U,{\cal F})$ be an instance of {\sc $d$-Hitting Set} where $U=\bigcup_{S\in{\cal F}}S$, and let $\alpha$ be an optimal solution to its classic LP that assigns only values strictly smaller than $\frac{1}{d-1}$. Then, $|U|\leq d\cdot\mathsf{frac}(I)$.
\end{lemma}

\begin{proof}
We first claim that every $S\in{\cal F}$ is a subset of $\mathsf{support}(\alpha)$. To this end, consider some set $S\in{\cal F}$. Then, because $\alpha$ is a solution, it satisfies $\sum_{u\in S}\alpha(x_u)\geq 1$. Targeting a contradiction, suppose that there exists $u'\in S\setminus\mathsf{support}(\alpha)$. Then, because $|S|\leq d$ and $\alpha$ assigns only values strictly smaller than $\frac{1}{d-1}$, we have that 
\[\displaystyle{\sum_{u\in S}\alpha(x_u) = \sum_{u\in S\setminus \{u'\}}\alpha(x_u)<\sum_{u\in S\setminus \{u'\}}\frac{1}{d-1}\leq 1},\]
which yields a contradiction.

We conclude that $\bigcup{\cal F}\subseteq \mathsf{support}(\alpha)$. By Theorem \ref{thm:support}, $|\mathsf{support}(\alpha)|\leq d\cdot \mathsf{frac}(I)$, and hence $|\bigcup{\cal F}|\leq d\cdot \mathsf{frac}(I)$. Because $U=\bigcup_{S\in{\cal F}}S$, the proof is complete.
\end{proof}

In particular, we now show this lemma yields the desired bound on the number of elements in the output instance of our kernelization algorithm:

\begin{lemma}\label{lem:dHSElemSize}
Let $I=(U,{\cal F})$ be an instance of {\sc $d$-Hitting Set}. Consider  a call to {\bf reduce} of the {\sc $d$-Hitting Set} element kernelization algorithm on input $I=(U,{\cal F})$ and whose output is $I'=(U',{\cal F}')$. Then, $|U'|\leq d\cdot\mathsf{frac}(I')$ and $|{\cal F}'|\leq(d\cdot\mathsf{frac}(I'))^d$.
\end{lemma}

\begin{proof}
Due to the condition to break the loop in {\bf reduce}, we have an instance $I'$ whose classic LP admits an optimal solution $\alpha'$ that assigns only values strictly smaller than $\frac{1}{d-1}$. Moreover, recall that $U'=\bigcup_{S\in{\cal F}'}S$.  So, by Lemma \ref{lem:dHSFinalElemSize}, $|U'|\leq d\cdot\mathsf{frac}(I')$. Clearly, this also implies that $|{\cal F}'|\leq {|U'|\choose d}\leq {d\cdot\mathsf{frac}(I')\choose d}\leq(d\cdot\mathsf{frac}(I))^d$.
\end{proof}

We now justify the approximation ratio of our kernelization algorithm. We remark that the particular way in which we phrase it, in particular distinguishing between the two items in its statement rather than only in its proof, is required for later purposes, as we explain before stating Theorem~\ref{thm:HSElementFrac}.

\begin{lemma}\label{lem:dHSElementApprox}
Let $I=(U,{\cal F})$ be an instance of {\sc $d$-Hitting Set}. Consider  a call to  {\bf lift} of the {\sc $d$-Hitting Set} element kernelization algorithm on input $I=(U,{\cal F}),I'=(U',{\cal F}'),S'$ and whose output is $S$. 
For any $0<\rho$, at least one of the following conditions holds:
\begin{enumerate}
\item  $|S|-|S'|\leq \rho\cdot\mathsf{opt}(I)$.
\item $\displaystyle{\frac{|S|}{\mathsf{opt}(I)}\leq d-\frac{\rho}{d-1}}$.
\end{enumerate}
Furthermore, $\displaystyle{\frac{|S|}{\mathsf{opt}(I)}\leq(d-\frac{d-1}{d})\frac{|S'|}{\mathsf{opt}(I')}}$.
\end{lemma}

\begin{proof}
We consider two cases, depending on $|H^\star|$.
\begin{enumerate}
\item First, suppose that $|H^\star|\leq \rho\cdot\mathsf{opt}(I)$. Then, because $|S|-|S'|=|H^\star|$, we directly have that $|S|-|S'|\leq \rho\cdot\mathsf{opt}(I)$.

\item Second, suppose that $|H^\star|\geq\rho\cdot\mathsf{opt}(I)$.  Let $t$ denote the number of iterations performed by the {\sc $d$-Hitting Set} element kernelization algorithm. For every $i\in\{1,2,\ldots,t-1\}$, observe that $\alpha_i|_{\{x_u: u\in U_{i+1}\}}$ is a solution to the classic LP corresponding to $I_{i+1}$, therefore $\displaystyle{\sum_{u\in U_{i+1}}\alpha_{i+1}(x_u)\leq \sum_{u\in U_{i+1}}\alpha_i(x_u)=\sum_{u\in U_i}\alpha_i(x_u)-\sum_{u\in H_i}\alpha_i(x_u)\leq \sum_{u\in U_i}\alpha_i(x_u)-\frac{1}{d-1}|H_i|}$.\\
 (Here, the last inequality follows since $\alpha_i(x_u)\geq \frac{1}{d-1}$ for every $u\in H_i$.)
So,
\[\begin{array}{ll}
\displaystyle{\sum_{u\in U_t}\alpha_t(x_u)} &\leq \displaystyle{\sum_{u\in U_{t-1}}\alpha_{t-1}(x_u)-\frac{1}{d-1}|H_{t-1}|}\\
&\leq \displaystyle{\sum_{u\in U_{t-2}}\alpha_{t-2}(x_u)-\frac{1}{d-1}|H_{t-1}|-\frac{1}{d-1}|H_{t-2}|}\\
& ...\\
&\leq \displaystyle{\sum_{u\in U_1}\alpha_1(x_u)-\frac{1}{d-1}|H_{t-1}|-\frac{1}{d-1}|H_{t-2}| - \ldots - \frac{1}{d-1}|H_{1}|}\\
& = \displaystyle{\sum_{u\in U_1}\alpha_1(x_u)-\frac{1}{d-1}|H^\star|}.
\end{array}\]
In particular, $\mathsf{frac}(I')\leq\mathsf{frac}(I)-\frac{1}{d-1}|H^\star|$. Moreover, by Lemma \ref{lem:dHSElemSize} and because $S'\subseteq U'$, we know that $|S'|\leq d\cdot \mathsf{frac}(I')$. So, 
\[\begin{array}{ll}
|S|=|S'|+|H^\star|&\leq d\cdot \mathsf{frac}(I’)+|H^\star|\\
&\leq d\cdot (\mathsf{frac}(I)-|H^\star|/(d-1))+|H^\star|\\
& \leq d\cdot \mathsf{opt}(I)-|H^\star|/(d-1)\\
& \leq (d-\rho/(d-1))\cdot \mathsf{opt}(I).
\end{array}\]
This directly implies that $\displaystyle{\frac{|S|}{\mathsf{opt}(I)}\leq d-\frac{\rho}{d-1}}$.
\end{enumerate}
This proves the first part of the lemma. For the second part, we choose $\rho=\frac{(d-1)^2}{d}$. Now, we show that in each of the aforementioned two cases,  $\displaystyle{\frac{|S|}{\mathsf{opt}(I)}\leq(d-\frac{d-1}{d})\frac{|S'|}{\mathsf{opt}(I')}}$. For the second case, this directly follows by substituting $\rho$ by $\frac{(d-1)^2}{d}$. So, in what follows, we only consider the first case, where $|S|-|S'|\leq \rho\cdot\mathsf{opt}(I)=\frac{(d-1)^2}{d}\cdot\mathsf{opt}(I)$, 
and hence $|S|\leq |S'|+\frac{(d-1)^2}{d}\cdot\mathsf{opt}(I)$. Then,
\[\begin{array}{ll}
\displaystyle{\frac{|S|}{\mathsf{opt}(I)}}& \leq \displaystyle{\frac{|S'|+\frac{(d-1)^2}{d}\cdot\mathsf{opt}(I)}{\mathsf{opt}(I)}}\\

&\leq\displaystyle{\frac{|S'|}{\mathsf{opt}(I')} + \frac{(d-1)^2}{d}}\\

&\leq\displaystyle{(1+ \frac{(d-1)^2}{d})\frac{|S'|}{\mathsf{opt}(I')}}
\end{array}\]
Here, the second inequality follows since $\mathsf{opt}(I')\leq \mathsf{opt}(I)$, and the third inequality follows since $|S'|\geq \mathsf{opt}(I')$. Now, observe that $1+\frac{(d-1)^2}{d}=\frac{d}{d}+\frac{d^2-2d+1}{d}=\frac{d^2-d+1}{d}=d-\frac{d-1}{d}$. So, indeed $\displaystyle{\frac{|S|}{\mathsf{opt}(I)}\leq(d-\frac{d-1}{d})\frac{|S'|}{\mathsf{opt}(I')}}$.
\end{proof}

We are now ready to prove the main theorem of this subsection. In particular, while we prove that our kernelization algorithm is a $(d-\frac{d-1}{d})$-approximate $d\cdot\mathsf{frac}$-element and $(d\cdot\mathsf{frac})^d$-set kernel, we also state that it is \good, and we should keep in mind that it also satisfies the two conditions in Lemma \ref{lem:dHSElementApprox}. In particular, we will need the two conditions in this lemma for the purpose of being able to compose it later: rather than incurring  a $(d-\frac{d-1}{d})$ multiplicative error, it can be used so that it either incurs an {\em (essentially) negligible additive} error, or returns a solution $S$ of approximation ratio better than $d$ (though not $(d-\frac{d-1}{d})$, but depending on how ``negligible'' the additive error in the first case should be) irrespective of the approximation ratio of the solution $S'$ given to it.  These conditions will be necessary for the correctness of our approximate kernelization protocol for {\sc $d$-Hitting Set} that is given in the next section.

\begin{theorem}\label{thm:HSElementFrac}
The {\sc $d$-Hitting Set} problem, parameterized by the fractional optimum of the classic LP, admits a $(d-\frac{d-1}{d})$-approximate $d\cdot\mathsf{frac}$-element and $(d\cdot\mathsf{frac})^d$-set kernel. Furthermore, it is  \good.
\end{theorem}

\begin{proof}
Clearly, the {\bf lift} procedure of the kernelization algorithm is performed in polynomial time. Further, the loop of the {\bf reduce} procedure can perform at most $|U|$ iterations before the one where it breaks (since each of them removes at least one element from the universe), and each is performed in polynomial time, so overall this procedure is performed in polynomial time. The bounds on the number of elements in the output as well as its size, along with the property of being \good, follow from Lemma \ref{lem:dHSElemSize}. Lastly, the approximation ratio follows from Lemma \ref{lem:dHSElementApprox}. This completes the proof.
\end{proof}

Because parameterization by the fractional optimum of the classic LP is lower bounded by parameterization by the optimum, and due to Lemma \ref{lem:lossyKerOptToK}, we have the following corollaries of Theorem \ref{thm:HSElementFrac}.

\begin{corollary}
The {\sc $d$-Hitting Set} problem, parameterized by the optimum, admits a $(d-\frac{d-1}{d})$-approximate $d\cdot\mathsf{opt}$-element $(d\cdot\mathsf{opt})^d$-set kernel.
\end{corollary}

\begin{corollary}\label{thm:HSElementK}
The {\sc $d$-Hitting Set} problem, parameterized by a bound $k$ on the solution size, admits a $(d-\frac{d-1}{d})$-approximate $\frac{d}{d-\frac{d-1}{d}}\cdot (k+1)$-element $(\frac{d}{d-\frac{d-1}{d}}\cdot (k+1))^d$-set kernel.
\end{corollary}

It is noteworthy that when $d=2$, in which {\sc $d$-Hitting Set} equals {\sc Vertex Cover}, we retrieve the classic result that {\sc Vertex Cover} admits a $1$-approximate (i.e., exact) $2k$-vertex kernelization algorithm~\cite{FominLSZ19}. This does not follow directly from the stated approximation ratio of $d-\frac{d-1}{d}$ (which equals $1\frac{1}{2}$ rather than $1$ when $d=2$). However, the argument used to prove the correctness of the classic result, that is, that there exists a solution that contains all vertices whose variables are assigned $1$, also implies for our kernel that it is exact (see, e.g., \cite{FominLSZ19}).
Thus, our theorem regarding {\sc $d$-Hitting Set} can be viewed as a generalization of this classic result.


\newcommand{\coeffConst}{\frac{1}{10}}
\newcommand{\coeffConstD}{\frac{1}{10d}}
\newcommand{\approxVC}{\frac{2}{\sqrt{10}-2}}

\section{A Pure $d'$-Approximate Kernelization Protocol for {\sc $d$-Hitting Set} of Almost Linear Call Size where $d'<d$}\label{sec:HSSize}

For the sake of clarity, we first give a warm-up example. Afterwards, we present our general result that is based on the approach presented by that warm-up example, non-trivial insights regarding how to apply that approach in a recursive manner, and critically also on Theorem~\ref{thm:support} (via Theorem \ref{thm:HSElementFrac}). Lastly, we present some further outlook by relating a method to prove the existence of a $(1+\epsilon)$-approximate kernelization protocol for {\sc Vertex Cover} to the non-existence of $(r,t)$-Ruzsa-Szemer\'{e}di graphs where $r$ is linear in $n$ (the number of vertices) and $t$ is ``large'', which is an open problem.

We will make use of a polynomial-time $d$-approximation algorithm for {\sc $d$-Hitting Set}:

\begin{proposition}[Folklore]\label{prop:approxHS}
The {\sc $d$-Hitting Set} problem admits a polynomial-time $d$-approximation algorithm.
\end{proposition}

\subsection{Warm-Up Example: A 1.721-Approximate Kernelization Protocol for {\sc Vertex Cover} of $2$ Rounds and Call Size $(2k)^{1.5}$}\label{sec:warmUp}

We start with a warm-up and, in a sense, toy example which exemplifies a main insight behind our more general result, that is, that essentially we may use the oracle to find a ``large subinstance'' that is ``sparse'', and hence which (with another oracle call), we can solve optimally.  We will make use of Theorem \ref{thm:support} (as to stay as close as possible to the proof of the more general result, where it is necessary), though here, as $d=2$, one can equally use the classic $1$-approximate $2k$-vertex kernel for {\sc Vertex Cover}~\cite{FominLSZ19}.

\begin{theorem}\label{thm:warmUp}
The {\sc Vertex Cover} problem, parameterized by the fractional optimum of the classic LP, admits a pure, having $2$ rounds, $\approxVC$-approximate\footnote{Note that $\approxVC\leq 1.721$.} (randomized)\footnote{Here, randomization means that we may fail to return a $(\approxVC)$-approximate solution (i.e., we may return a ``worse'' solution), but we must succeed with probability, say, at least $9/10$. It should be clear that the success probability can be boosted to any constant arbitrarily close to $1$.} kernelization protocol with call size $2\mathsf{frac}+2(2\mathsf{frac})^{1.5}$ (where the number of edges is at most $2(2\mathsf{frac})^{1.5}$).
\end{theorem}

\begin{proof}
We first describe the algorithm. To this end, consider some input $\widehat{I}=(\widehat{U},\widehat{\cal F})$ (in terms of graphs, $\widehat{U}$ is the vertex set and $\widehat{\cal F}$ is the edge set of the input graph).\footnote{We represent the input using a universe and sets so that it will  resemble our more general protocol more.} Then:
\begin{enumerate}
\item Call the {\bf reduce} procedure of the algorithm in Theorem \ref{thm:HSElementFrac} on $\widehat{I}$ to obtain a new instance $I=(U,{\cal F})$ where $|U|\leq 2\mathsf{frac}(I)$. (Recall that when $d=2$, this algorithm is exact.)\footnote{See the discussion at the end of Section \ref{sec:HSElement}.}

\item Let $0<\nu<1$ (analogous to $\frac{\mu}{2}$ in the general result) be a fixed constant that will be determined later.

\item Sample ${\cal F}_1$ from ${\cal F}$ as follows: Insert each set $S\in{\cal F}$ to ${\cal F}_1$ independently at random with probability $p_1=\displaystyle{\frac{1}{(2\mathsf{frac}(I))^{0.5}}}$.

\item If $|{\cal F}_1|>2p_1|{\cal F}|$, then let $S$ be an arbitrary solution to $I$, and proceed directly to Step \ref{step:lift1}. [\#Failure]

\item Call the oracle on $(U,{\cal F}_1)$, and let $S_1$ denote its output. 

	\item\label{step:success1} If $|S_1|\geq \nu|U|$, then let $S=U$, and proceed directly to Step \ref{step:lift1}.  [\#Success]

\item Let $U_1=U\setminus S_1$ and ${\cal T}_1=\{S\in{\cal F}: S\subseteq U_1\}$.

\item If $|{\cal T}_1|>2(2\mathsf{frac}(I))^{1.5}$, then let $S$ be an arbitrary solution to $I$, and proceed directly to Step \ref{step:lift1}. [\#Failure]

\item Call the oracle on $I'=(U_1,{\cal T}_1)$, and let $S_2$ denote its output. 

\item\label{step:success2} Let $S'=S_2\cup S_1$ and $S''=U_1\cup T$ where $T$ is a $2$-approximate solution to $\widetilde{I}=(S_1,\{S\in{\cal F}: S\subseteq S_1\})$ (computed using Proposition \ref{prop:approxHS}). Let $S$ be a minimum-sized set among $S'$ and $S''$. [\#Success]

\item\label{step:lift1} Call the {\bf lift} procedure of the algorithm in Theorem \ref{thm:HSElementFrac} on $\widehat{I},I,S$ to obtain a solution $\widehat{S}$ to $\widehat{I}$. Output $\widehat{S}$.
\end{enumerate}

Clearly, the algorithm runs in polynomial time, and only two oracle calls are performed. Further, when we call the oracle on $(U,{\cal F}_1)$,  $|{\cal F}_1|\leq 2p_1|{\cal F}|\leq 2\cdot \displaystyle{\frac{1}{(2\mathsf{frac}(I))^{0.5}}}\cdot (2\mathsf{frac}(I))^{2}=2(2\mathsf{frac}(I))^{1.5}$ (due to {\bf reduce}). Thus, each oracle call is performed on an instance with at most $2\mathsf{frac}(I)$ vertices (as $|U|\leq 2\mathsf{frac}(I)$ due to {\bf reduce}) and $2(2\mathsf{frac}(I))^{1.5}$ edges, and since $\mathsf{frac}(I)\leq\mathsf{frac}(\widehat{I})$, the statement in the lemma regarding the call size is satisfied.

We now consider the probability of failure. By Chernoff bound (Proposition \ref{prop:chernoff}), the probability that $|{\cal F}_1|>2p_1|{\cal F}|$ is at most $e^{-\frac{p_1|{\cal F}|}{3}}$. Further, by union bound, the probability that there exists a subset $U'\subseteq U$ such that ${\cal F}_1\cap{\cal F}_{U'}=\emptyset$ (where ${\cal F}_{U'}=\{S\in{\cal F}: S\subseteq U'\}$) under the assumption that $|{\cal F}_{U'}|>2(2\mathsf{frac}(I))^{1.5}$ is at most $2^{2\mathsf{frac}(I)}\cdot (1-p_1)^{2(2\mathsf{frac}(I))^{1.5}}=2^{2\mathsf{frac}(I)}\cdot (1-\frac{1}{(2\mathsf{frac}(I))^{0.5}})^{2(2\mathsf{frac}(I))^{1.5}}\leq 2^{2\mathsf{frac}(I)}\cdot e^{-4\mathsf{frac}(I)}$. Thus, by union bound, under the implicit supposition that $\mathsf{frac}$ (and $|{\cal F}|$) is a large enough constant (e.g., $10$),\footnote{Otherwise, the instance can be solved optimally in polynomial time using brute-force.} the probability that at least one of the events in the steps marked by ``failure'' occurs is at most $1/10$. Notice that if these events occur, $S$ is a solution. Further, we now claim that if these events do not occur, then we compute a set $S$ that is a solution to $I$ and, furthermore, it is $\approxVC$-approximate. Then, by the correctness of {\bf lift} (in particular, since the kernelization algorithm in Theorem \ref{thm:HSElementFrac} is $1$-approximate, that is, exact, for $d=2$), this will conclude the proof. For this purpose, we have the following case distinction, where $\beta$ is the approximation ratio of the oracle.

First, suppose that $S$ is computed in Step \ref{step:success1}. Then, $|S_1|\geq\nu|U|$ and $S=U$. Clearly, $S$ is a solution to $I$. Because $S_1$ is a $\beta$-approximate solution to $(U,{\cal F}_1)$, which is a subinstance of $(U,{\cal F})$, this means that $\mathsf{opt}(I)\geq \frac{\nu}{\beta}|U|$. So, in this case, the approximation ratio is $\frac{|S|}{\mathsf{opt}(I)}\leq\frac{|U|}{\frac{\nu}{\beta}|U|}=\beta\frac{1}{\nu}$. 

Second, suppose that $S$ is computed in Step \ref{step:success2}. Then, $|S_1|<\nu|U|$. On the one hand, because $S_2$ is a solution to $I'=(U_1,{\cal T}_1)$, and, as ${\cal T}_1=\{S\in{\cal F}:S\subseteq U_1\}$, every set in ${\cal F}\setminus{\cal T}_1$ contains at least one vertex from $U\setminus U_1=S_1$,  we have that $S'=S_2\cup S_1$ is a solution to $I$. Further, since $S_2$ is a $\beta$-approximate solution to $I'$, $|S'|\leq \beta\mathsf{opt}(I')+|S_1|$. On the other hand, because $T$ is a solution to $\widetilde{I}=(S_1,\{S\in{\cal F}: S\subseteq S_1\})$, and every set in ${\cal F}\setminus\{S\in{\cal F}: S\subseteq S_1\}$ contains at least one vertex from $U_1$,  we have that $S''=U_1\cup T$ is also a solution to $I$. Further, because $T$ is a $2$-approximate solution to $\widetilde{I}$, $|S''|\leq 2\mathsf{opt}(\widetilde{I})+|U_1|=2\mathsf{opt}(\widetilde{I})+|U|-|S_1|$.

Consider some optimal solution $S^\star$ to $I$. Then, $S^\star\setminus S_1$ is a solution to $I'$, and $S^\star\cap S_1$ is a solution to $\widetilde{I}$, which means that $\mathsf{opt}(I')\leq |S^\star\setminus S_1|$ and $\mathsf{opt}(\widetilde{I})\leq |S^\star\cap S_1|$. So, denoting $\lambda=\frac{|S^\star\cap S_1|}{|S_1|}$ ($0\leq\lambda\leq 1$) and $\rho=\frac{|S_1|}{|U|}$ ($0\leq\rho<\nu$), we know that
\begin{itemize}
\item $|S'|\leq \beta|S^\star\setminus S_1|+|S_1|=\beta\mathsf{opt}(I)-\beta|S^\star\cap S_1|+|S_1|=\beta\mathsf{opt}(I)+(1-\beta\lambda)\rho|U|\leq (\beta+2\rho-2\beta\lambda\rho)\mathsf{opt}(I)$.
\item $|S''|\leq 2|S^\star\cap S_1|+|U|-|S_1|=(1+2\lambda\rho-\rho)|U|\leq (2+4\lambda\rho-2\rho)\mathsf{opt}(I)$.
\end{itemize}
As $\lambda$ grows larger, the first term becomes better, and as it grows smaller, the second term is better. So, the worst case is such that equality is attained when $\lambda=\frac{\beta+4\rho-2}{2(2+\beta)\rho}$. Then, the approximation ratio is $2+4\frac{\beta+4\rho-2}{2(2+\beta)\rho}\rho-2\rho=2+\frac{2(\beta+4\rho-2)}{2+\beta}-2\rho=2-\frac{4-2\beta}{2+\beta}+(\frac{4-2\beta}{2+\beta})\rho$. When $\beta\geq 2$, the correctness of the approximation ratio is trivial, since then even returning all of $U$ is a $\beta$-approximation. So, suppose that $\beta<2$. Then, the aforementioned function grows larger as $\rho$ grows larger (since when $\beta<2$, its coefficient is positive), and as $\rho<\nu$, an upper bound on the maximum is $2+(\frac{4-2\beta}{2+\beta})\nu-\frac{4-2\beta}{2+\beta}$. Now, we fix $\nu$ such that $2+(\frac{4-2\beta}{2+\beta})\nu-\frac{4-2\beta}{2+\beta}=\frac{\beta}{\nu}$ when $\beta=1$. So, we require $\frac{4}{3}+\frac{2}{3}\nu=\frac{1}{\nu}$, that is, $2\nu^2+4\nu-3=0$, which is satisfied when $\nu=\frac{\sqrt{10}}{2}-1$. Then, in the first case, the approximation ratio is at most $\beta\frac{1}{\nu}=\beta\approxVC$ as required. In the second case, the approximation ratio is also $2+(\frac{4-2\beta}{2+\beta})\nu-\frac{4-2\beta}{2+\beta}\leq \beta(\frac{4}{3}+\frac{2}{3}\nu)=\beta\approxVC$ as required. This completes the proof.
\end{proof}

\begin{corollary}
The {\sc Vertex Cover} problem, parameterized by the optimum, admits a pure, having $2$ rounds, $\approxVC$-approximate (randomized) kernelization protocol with call size $2\mathsf{opt}+2(2\mathsf{opt})^{1.5}$ (where the number of edges is at most $2(2\mathsf{opt})^{1.5}$).
\end{corollary}

\subsection{Generalization to Almost Linear Call Size and $d\geq 2$}

A critical part of our algorithm is Theorem \ref{thm:HSElementFrac}. First, after calling its algorithm to reduce the number of elements, there will only be $2^{d\mathsf{frac}}$ many subsets of $U$ such that, if the instance induced by them is not ``sparse enough'' (where the definition of sparse enough becomes stricter and stricter as the execution of our algorithm proceeds), then with high probability we will ``hit'' at least one of their sets when using an oracle call.
Further, Theorem \ref{thm:HSElementFrac} will be used to prove that, after calling its algorithm to reduce the number of elements, once we find a ``sufficiently'' large (linear in $k$) subset of $U$ along with a solution to the instance induced by that subset that is large compared to its size (in particular, consisting of more that a fraction of $1/d$ of its elements), we are essentially done. Our algorithm will repeatedly try to find subsets as mentioned above, while, if it fails at every step, it eventually arrives at a ``sufficiently'' large (linear in $k$) subset of $U$ such that it can optimally solve the instance induced by that subset. 

\begin{theorem}\label{thm:HSSize}
For any fixed $\epsilon>0$, the {\sc $d$-Hitting Set} problem, parameterized by the fractional optimum of the classic LP, admits a pure $d(1-h(d,\epsilon))$-approximate (randomized)\footnote{Here, randomization means that we may fail to return a $(d-h(d,\epsilon))$-approximate solution (i.e., we may return a ``worse'' solution), but we must succeed with probability, say, at least $9/10$. It should be clear that the success probability can be boosted to any constant arbitrarily close to $1$.} kernelization protocol with call size $d\cdot\mathsf{frac}+2^{\frac{d}{\epsilon}}(d\cdot\mathsf{frac})^{1+\epsilon}$ (where the number of sets is at most $(d\cdot\mathsf{frac})^{1+\epsilon}$) where $h(d,\epsilon)=\displaystyle{\coeffConstD(\frac{1}{4})^{\frac{d}{\epsilon}}}$ is a fixed positive constant that depends only on $d,\epsilon$.\footnote{We remark that we preferred to simplify the algorithm and its analysis rather than to optimize $h(d,\epsilon)$ (in fact, the same algorithms with slightly more careful analysis already yields a much better yet ``uglier'' constant). In particular, our approximation ratio is a {\em fixed constant} (under the assumption that $d,\epsilon$ are fixed) strictly smaller than $d$.}
\end{theorem}

\begin{proof}
We first describe the algorithm. To this end, consider some input $\widehat{I}=(\widehat{U},\widehat{\cal F})$. Then:
\begin{enumerate}
\item Call the {\bf reduce} procedure of the algorithm in Theorem \ref{thm:HSElementFrac} on $I$ to obtain a new instance $I=(U,{\cal F})$ where $|U|\leq d\cdot\mathsf{frac}(I)$.

\item Denote $\mu=\mu(d)=\frac{d+1}{2}$, and $\tau=\tau(d,\epsilon)=\frac{1}{\epsilon}(d-1)$.

\item Initialize $U_0=U$ and ${\cal T}_0={\cal F}$.

\item For $i=1,2,\ldots,\tau$:
	\begin{enumerate}
	\item Sample ${\cal F}_i$ from ${\cal T}_{i-1}$ as follows: Insert each set $S\in{\cal T}_{i-1}$ to ${\cal F}_i$ independently at random with probability $p_i=\displaystyle{\frac{1}{(d\cdot\mathsf{frac}(I))^{d-1-i\cdot\epsilon}}}$.
	\item\label{step:failure} If $|{\cal F}_i|>2^i(d\cdot\mathsf{frac}(I))^{1+\epsilon}$, then let $S$ be an arbitrary solution to $I$, and proceed directly to Step \ref{step:lift}. [\#Failure]
	\item Call the oracle on $(U_{i-1},{\cal F}_i)$, and let $S_i$ denote its output. [\#We will verify in the proof that all calls are done with at most $(d\cdot\mathsf{frac}(\widehat{I}))^{1+\epsilon}$ sets.]
	\item\label{step:checkLarge} If $|S_i|\geq \frac{\mu}{d}|U_{i-1}|$, then:
		\begin{enumerate}
		\item Call the algorithm in Proposition \ref{prop:approxHS} on $(U\setminus U_{i-1},\{S\in{\cal F}: S\subseteq U\setminus U_{i-1}\})$, and let $T$ denote its output.
		\item\label{step:existLarge} Let $S=T\cup U_{i-1}$ and proceed directly to Step \ref{step:lift}.  [\#Success]
		\end{enumerate}
		\item\label{step:Ti} Let $U_i=U_{i-1}\setminus S_i$ and ${\cal T}_{i}=\{S\in{\cal T}_{i-1}: S\subseteq U_i\}$.
	\end{enumerate}

\item\label{step:existSmall} Let $S=S_{\tau}\cup (U\setminus U_{\tau-1})$. [\#Success]

\item\label{step:lift} Call the {\bf lift} procedure of the algorithm in Theorem \ref{thm:HSElementFrac} on $\widehat{I},I,S$ to obtain a solution $\widehat{S}$ to $\widehat{I}$. Output $\widehat{S}$.
\end{enumerate}

Clearly, the algorithm runs in polynomial time. Further, each oracle call has at most $d\cdot\mathsf{frac}(I)$ many elements. We first verify that it also has at most $(d\cdot\mathsf{frac}(I))^{1+\epsilon}$ sets. To this end, we have two preliminary claims.

\begin{claim}\label{claim:callBoundHelper1}
For all $i=1,2,\ldots,\tau$, if the algorithm reaches iteration $i$ and $|{\cal T}_{i-1}|\leq 2^{i-1}(d\cdot\mathsf{frac}(I))^{d-(i-1)\epsilon}$, then $|{\cal F}_i|\leq 2^i(d\cdot\mathsf{frac}(I))^{1+\epsilon}$ with probability at least $1-e^{-\frac{2^{i-1}(d\cdot\mathsf{frac}(I))^{1+\epsilon}}{3}}$.
\end{claim}

\begin{proof}
Observe that the expected size of ${\cal F}_i$ is:
\[\begin{array}{ll}
\mathsf{E}[|{\cal F}_i|]& =|{\cal T}_{i-1}|\cdot p_i\\
& \leq \displaystyle{2^{i-1}(d\cdot\mathsf{frac}(I))^{d-(i-1)\epsilon} \cdot \frac{1}{(d\cdot\mathsf{frac})^{d-1-i\cdot\epsilon}}}\\
& = 2^{i-1}(d\cdot\mathsf{frac}(I))^{1+\epsilon}.
\end{array}\]
Thus, Chernoff bound (Proposition \ref{prop:chernoff}) implies that
\[\mathsf{Prob}[|{\cal F}_i|>2^i(d\cdot\mathsf{frac}(I))^{1+\epsilon}]\leq e^{-\frac{2^{i-1}(d\cdot\mathsf{frac}(I))^{1+\epsilon}}{3}}.\]
This completes the proof of the claim.
\end{proof}

\begin{claim}\label{claim:callBoundHelper2}
For all $i=1,2,\ldots,\tau$, if the algorithm reaches Step \ref{step:Ti} in iteration $i$ and $|{\cal T}_{i-1}|\leq 2^{i-1}(d\cdot\mathsf{frac}(I))^{d-(i-1)\epsilon}$, then $|{\cal T}_i|\leq 2^i(d\cdot\mathsf{frac}(I))^{d-i\epsilon}$ with probability at least $1-\displaystyle{(\frac{2}{e})^{2^i\cdot d\cdot\mathsf{frac}(I)}}$.
\end{claim}

\begin{proof}
Consider some iteration $i\in\{1,2,\ldots,\tau\}$, and suppose that the algorithm reaches Step \ref{step:Ti} iteration $i$. Hence, $|{\cal F}_i|\leq 2^i(d\cdot\mathsf{frac}(I))^{1+\epsilon}$. Consider some subfamily ${\cal T}'\subseteq {\cal T}_{i-1}$ such that $|{\cal T}'|>2^i(d\cdot\mathsf{frac}(I))^{d-i\epsilon}$.
Then,
\[
\begin{array}{ll}
\mathsf{Prob}({\cal T}\cap{\cal F}_i=\emptyset) &=(1-p_i)^{|{\cal T}'|}\\
& \leq \displaystyle{(1-\frac{1}{(d\cdot\mathsf{frac})^{d-1-i\cdot\epsilon}})^{2^i(d\cdot\mathsf{frac}(I))^{d-i\epsilon}}}\\
& \leq \displaystyle{e^{-2^i\cdot d\cdot\mathsf{frac}(I)}}.
\end{array}\]
Because there exist at most $2^{d\cdot\mathsf{frac}(I)}$ subsets of $U_{i-1}$, union bound implies that the probability that there exists $U'\subseteq U_{i-1}$ such that the subfamily $\{S\in{\cal T}_{i-1}: S\subseteq U'\}$ is of size larger than $2^i(d\cdot\mathsf{frac}(I))^{d-i\epsilon}$ and has empty intersection with ${\cal F}_i$ is at most $\displaystyle{(\frac{2}{e})^{2^i\cdot d\cdot\mathsf{frac}(I)}}$. Recall that ${\cal T}_{i}=\{S\in{\cal T}_{i-1}: S\subseteq U_i\}$ and note that ${\cal T}_i$ has empty intersection with ${\cal F}_i$ because $S_i=U_{i-1}\setminus U_i$ is a solution to $(U_i,{\cal F}_i)$ (by the correctness of the oracle). This completes the proof of the claim.
\end{proof}

We now prove the desired bound on each call size, based on Claims \ref{claim:callBoundHelper1} and \ref{claim:callBoundHelper2}.

\begin{claim}\label{claim:callBound}
The following statement holds with probability at least $9/10$: For all $i=1,2,\ldots,\tau$, if the algorithm reaches iteration $i$ and calls the oracle, then $|{\cal F}_i|\leq 2^i(d\cdot\mathsf{frac}(I))^{1+\epsilon}$ and the algorithm does not exit the loop in Step \ref{step:failure}.
\end{claim}

\begin{proof}
We claim that for every $j\in\{0,1,\ldots,\tau\}$, the following holds with probability at least $1-\sum_{i=1}^j\displaystyle{(\frac{2}{e})^{2^i\cdot d\cdot\mathsf{frac}(I)}}$: for every $i\in\{0,1,\ldots,j\}$ such that the algorithm reaches Step \ref{step:Ti} in iteration $i$ (when $i=0$, we mean the initialization), $|{\cal T}_i|\leq 2^i(d\cdot\mathsf{frac}(I))^{d-i\epsilon}$. The proof is by induction on $j$. At the basis, where $j=0$, ${\cal T}_0={\cal F}$, and hence due to {\bf  reduce}, with probability $1$, $|{\cal T}|\leq (d\cdot\mathsf{frac}(I))^d$. Now, suppose that the claim is true for $j-1$, and let us prove it for $j$. By the inductive hypothesis, with probability at least $1-\sum_{i=1}^{j-1}\displaystyle{(\frac{2}{e})^{2^i\cdot d\cdot\mathsf{frac}(I)}}$, the following holds: for every $i\in\{0,1,\ldots,j-1\}$ such that the algorithm reaches Step \ref{step:Ti} in iteration $i$, $|{\cal T}_i|\leq 2^i(d\cdot\mathsf{frac}(I))^{d-i\epsilon}$. Now, if the algorithm further reaches Step \ref{step:Ti} in iteration $j$, Claim \ref{claim:callBoundHelper2} implies that $|{\cal T}_j|\leq 2^j(d\cdot\mathsf{frac}(I))^{d-j\epsilon}$ with probability at least $1-\displaystyle{(\frac{2}{e})^{2^j\cdot d\cdot\mathsf{frac}(I)}}$. So, by union bound, the claim for $j$ it true.

In particular, by setting $j=\tau$, we have that with probability at least $1-\sum_{i=1}^\tau\displaystyle{(\frac{2}{e})^{2^i\cdot d\cdot\mathsf{frac}(I)}}$, the following holds: for every $i\in\{0,1,\ldots,\tau\}$ such that the algorithm reaches Step \ref{step:Ti} in iteration $i$ (when $i=0$, we mean the initialization), $|{\cal T}_i|\leq 2^i(d\cdot\mathsf{frac}(I))^{d-i\epsilon}$. However, by Claim \ref{claim:callBoundHelper1} and union bound, this directly extends to the following statement:  with probability at least $1-\sum_{i=1}^\tau\displaystyle{(\frac{2}{e})^{2^i\cdot d\cdot\mathsf{frac}(I)}}-\sum_{i=1}^\tau e^{-\frac{2^{i-1}(d\cdot\mathsf{frac}(I))^{1+\epsilon}}{3}}$, the following holds: for every $i\in\{0,1,\ldots,\tau\}$ such that the algorithm reaches iteration $i$ and calls the oracle, then $|{\cal F}_i|\leq 2^i(d\cdot\mathsf{frac}(I))^{1+\epsilon}$ and the algorithm does not exit the loop in Step \ref{step:failure}. Now, observe that
\[\begin{array}{ll}
\displaystyle{\sum_{i=1}^\tau\displaystyle{(\frac{e}{2})^{-2^i\cdot d\cdot\mathsf{frac}(I)}}+\sum_{i=1}^\tau e^{-\frac{2^{i-1}(d\cdot\mathsf{frac}(I))^{1+\epsilon}}{3}}} & \leq \tau\cdot((\frac{e}{2})^{-d\cdot\mathsf{frac}(I)}+e^{-\frac{(d\cdot\mathsf{frac}(I))^{1+\epsilon}}{3}})\\

&\leq 2\tau\cdot (\frac{e}{2})^{-\frac{(d\cdot\mathsf{frac}(I))^{1+\epsilon}}{3}} \leq \frac{1}{10}.
\end{array}\]
Here, the last inequality follows by assuming that $\mathsf{frac}(I)$ is large enough (to ensure that the inequality is satisfied) compared to $d,\epsilon$. Indeed, if this is not the case, then $\mathsf{frac}(I)$ (and hence also $\mathsf{opt}(I)$, because it bounded by $d\cdot\mathsf{frac}(I)$) is a fixed constant (that depends only on $d,\epsilon$), and hence the problem can just be a-priori solved in polynomial time by, e.g., brute force search. We thus conclude that the failure probability is at most $1/10$, which completes the proof of the claim.
\end{proof}

Let $\beta\geq 1$ denote the approximation ratio of the oracle. We now turn to analyze the approximation ratio. Towards that, we present a lower bound on the size of each universe $U_i$.

\begin{claim}\label{claim:universeLarge}
For all $i=1,2,\ldots,\tau$, if the algorithm reaches iteration $i$ and computes $U_i$, then $|U_i|\geq \displaystyle{(1-\frac{\mu}{d})^{i}|U|\geq (\frac{1}{4})^{\frac{1}{\epsilon}d}|U|}$. 
\end{claim}

\begin{proof}
We first claim that for all $i\in\{1,2,\ldots,\tau\}$, if the algorithm reaches iteration $i$ and computes $U_i$, then $|U_i|\geq \displaystyle{(1-\frac{\mu}{d})^{i}|U|}$. The proof is by induction on $i$, where we let $i=0$ be the basis.  Then, in the basis, $U_0=U$ and the claim trivially holds. Now, suppose that the claim holds for $i-1$, and let us prove it for $i$. By the inductive hypothesis, $|U_{i-1}|\geq (1-\frac{\mu}{d})^{i-1}|U|$. Further, by the definition of $U_i$, $U_i=U_{i-1}\setminus S_i$, and as the algorithm reaches the computation of $U_i$, $|S_i|<\frac{\mu}{d}|U_{i-1}|$. Thus, we have that
\[|U_i|\geq |U_{i-1}|-|S_i|> (1-\frac{\mu}{d})|U_{i-1}|\geq (1-\frac{\mu}{d})^{i}|U|.\]
Hence, the claim holds for $i$, and therefore our (sub)claim holds.

Lastly, observe that for all $i\in\{1,2,\ldots,\tau\}$, $|U_i|\geq |U_\tau|$. Moreover, due to our (sub)claim and substitution of $\tau$ and $\mu$, and because $\frac{x+1}{2x}\leq \frac{3}{4}$ for all $x\geq 2$ (the maximum is achieved when $x=2$), we have that
\[|U_\tau|\geq (1-\frac{\mu}{d})^{\tau}|U|=(1-\frac{d+1}{2d})^{\frac{1}{\epsilon}(d-1)}|U|\geq (\frac{1}{4})^{\frac{1}{\epsilon}d}|U|.\]
This completes the proof of the claim.
\end{proof}

Now, having the property that each universe $U_i$ is ``large enough'', we  argue that if $S$ is computed in Step \ref{step:existLarge}, then it is a solution of the approximation ratio $d(1-d\cdot h(d,\epsilon))$.

\begin{claim}\label{claim:exitLarge}
For all $i=1,2,\ldots,\tau$, if the algorithm reaches iteration $i$ and Step \ref{step:existLarge} of that iteration, then $S$  is a solution to $I$ such that $\frac{|S|}{\mathsf{opt}(I)}\leq \beta d(1-d\cdot h(d,\epsilon))$.
\end{claim}

\begin{proof}
Let $i\in\{1,2,\ldots,\tau\}$ such that the algorithm reaches iteration $i$ and Step \ref{step:existLarge} of that iteration. Then, $|S_i|\geq\frac{\mu}{d}|U_{i-1}|$ and $S=T\cup U_{i-1}$ (I). Let $S^\star$ be an optimal solution to $I$, so $|S^\star|=\mathsf{opt}(I)$ (II). Consider the following subinstances of $I$:
\begin{itemize}
\item $I'=(U_{i-1},{\cal F}_i)$. Because $S_i$ is a $\beta$-approximate solution to $I'$, we have that $\mathsf{opt}(I')\geq \frac{|S_i|}{\beta}\geq\frac{\mu}{\beta d}|U_{i-1}|$.
\item $I''=(U_{i-1},\{S\in{\cal F}: S\subseteq U_{i-1}\}))$. Because $I'$ is a subinstance of $I''$, we have that $\mathsf{opt}(I'')\geq\mathsf{opt}(I')$, and hence $\mathsf{opt}(I'')\geq \frac{\mu}{\beta d}|U_{i-1}|$. In particular, since $S^\star\cap U_{i-1}$ is a solution to $I''$, we have that $|S^\star\cap U_{i-1}|\geq \frac{\mu}{\beta d}|U_{i-1}|$. This has two consequences: first, $|U_{i-1}|\leq \frac{\beta d}{\mu}|S^\star\cap U_{i-1}|$ (III); second, due to Claim \ref{claim:universeLarge}, $|S^\star\cap U_{i-1}|\geq \frac{\mu}{\beta d}|U_{i-1}|\geq \frac{\mu}{\beta d}(\frac{1}{4})^{\frac{1}{\epsilon}d}|U|$ (IV).
\item $I'''=(U\setminus U_{i-1},\{S\in{\cal F}: S\subseteq U\setminus U_{i-1}\})$. Due to Proposition \ref{prop:approxHS}, $T$ is a solution to $I'''$ such that $|T|\leq d\cdot\mathsf{opt}(I''')$. Note that all sets in ${\cal F}$ that do not occur in this instance have non-empty intersection with $U_{i-1}$, and hence $S$ is a solution to $I$. Further, $S^\star\setminus U_{i-1}$ is a solution to $I'''$, and hence $|S^\star\setminus U_{i-1}|\geq \mathsf{opt}(I''')$. Thus, $|T|\leq d|S^\star\setminus U_{i-1}|$ (V). 
\end{itemize}
So, we have proved that $S$ is a solution to $I$, and we have that
\[\begin{array}{lll}
\smallskip
|S| & = |T|+|U_{i-1}| & [\mbox{(I)}]\\

\smallskip
& \leq d|S^\star\setminus U_{i-1}| + \frac{\beta d}{\mu}|S^\star\cap U_{i-1}| & [\mbox{(III)+(V)}]\\

\smallskip
& = d(|S^\star| - (1-\frac{\beta}{\mu})|S^\star\cap U_{i-1}|) & \\

\smallskip
& \leq d(|S^\star| - (1-\frac{\beta}{\mu})\frac{\mu}{\beta d}(\frac{1}{4})^{\frac{1}{\epsilon}d}|U|) & [\mbox{(IV)}]\\

\smallskip
& = d(\mathsf{opt}(I) - (1-\frac{\beta}{\mu})\frac{\mu}{\beta d}(\frac{1}{4})^{\frac{1}{\epsilon}d}|U|) & [\mbox{(II)}]\\

\smallskip
& \leq d(\mathsf{opt}(I) - (1-\frac{\beta}{\mu})\frac{\mu}{\beta d}(\frac{1}{4})^{\frac{1}{\epsilon}d}d\cdot\mathsf{frac}(I)) & [\mbox{Application of {\bf reduce}}] \\

\smallskip
& \leq \left(1-(1-\frac{\beta}{\mu})\frac{\mu}{\beta}(\frac{1}{4})^{\frac{1}{\epsilon}d}\right)d\cdot\mathsf{opt}(I)& \\

& = \left(\frac{1}{\beta}-(\frac{d+1}{2\beta^2}-\frac{1}{\beta})(\frac{1}{4})^{\frac{d}{\epsilon}}\right)\beta d\cdot\mathsf{opt}(I). &
\end{array}\]
Hence, $\frac{|S|}{\mathsf{opt}(I)}\leq \left(\frac{1}{\beta}-(\frac{d+1}{2\beta^2}-\frac{1}{\beta})(\frac{1}{4})^{\frac{d}{\epsilon}}\right)\beta d$. So, because $h(d,\epsilon)=\coeffConstD(\frac{1}{4})^{\frac{d}{\epsilon}}$, to conclude that $\frac{|S|}{\mathsf{opt}(I)}\leq  \beta d(1-d\cdot h(d,\epsilon))$, it suffices to prove that $\frac{1}{\beta}-(\frac{d+1}{2\beta^2}-\frac{1}{\beta})(\frac{1}{4})^{\frac{d}{\epsilon}}\leq 1-\coeffConst(\frac{1}{4})^{\frac{d}{\epsilon}}$. For this, we have the following case distinction.
\begin{itemize}
\item Suppose that $\beta\geq \frac{10}{9}$. Then, $\frac{1}{\beta}-(\frac{d+1}{2\beta^2}-\frac{1}{\beta})(\frac{1}{4})^{\frac{d}{\epsilon}}\leq \frac{1}{\beta}\leq \frac{9}{10}\leq 1-\coeffConst(\frac{1}{4})^{\frac{d}{\epsilon}}$.

\item Suppose that $\beta\leq \frac{10}{9}$. As $\frac{1}{\beta}\leq 1$, it suffices to prove that $\frac{d+1}{2\beta^2}-\frac{1}{\beta}\geq \coeffConst$, and as $d\geq 2$, it further suffices to prove that $\frac{3}{2\beta^2}-\frac{1}{\beta}\geq \coeffConst$. Because $\beta\leq \frac{10}{9}$, we have that $\frac{3}{2\beta^2}-\frac{1}{\beta}\geq  \frac{3}{2(\frac{10}{9})^2}-\frac{9}{10}\geq \coeffConst$.
\end{itemize}
This completes the proof.
\end{proof}

Further, we argue that if $S$ is computed in Step \ref{step:existSmall}, then also it is a solution of this approximation ratio. Towards that, we have the following trivial claim.

\begin{claim}\label{claim:exitSmallHelper}
For all $i=1,2,\ldots,\tau$, if the algorithm reaches iteration $i$ and computes ${\cal T}_i$, then ${\cal T}_i=\{S\in{\cal F}: S\subseteq U_i\}$.
\end{claim}

\begin{proof}
The proof is by induction on $i$ (where we use $i=0$ as basis). When $i=0$, $U_0=U$ and ${\cal T}_0={\cal F}$, thus the claim trivially holds. Now, suppose that it holds for $i-1$, and let us prove it for $i$. By the inductive hypothesis and the definition of ${\cal T}_i$, we have that 
\[{\cal T}_i=\{S\in{\cal T}_{i-1}: S\subseteq U_i\}=\{S\in{\cal F}: S\subseteq U_i\}.\]
This completes the proof of the claim.
\end{proof}

We now present the promised claim.

\begin{claim}\label{claim:exitSmall}
For all $i=1,2,\ldots,\tau$, if the algorithm reaches Step \ref{step:existSmall}, then $S$  is a solution to $I$ such that $\frac{|S|}{\mathsf{opt}(I)}\leq  d(1-d\cdot h(d,\epsilon))$.
\end{claim}

\begin{proof}
In this case, $S=S_{\tau}\cup(U\setminus U_{\tau-1})$. We first argue that $S$ is a solution to $I$. To this end, notice that $p_{\tau}=1$, so ${\cal F}_\tau={\cal T}_{\tau-1}$. This means, by the correctness of the oracle, that $S_\tau$ is a solution to $(U_{\tau-1},{\cal T}_{\tau-1})$. That is, it has non-empty intersection with every set in ${\cal T}_{\tau-1}$.  By Claim \ref{claim:exitSmallHelper}, ${\cal T}_{\tau-1}=\{S\in{\cal F}: S\subseteq U_{\tau-1}\}$, so $U\setminus U_{\tau-1}$ has non-empty intersection with every set in ${\cal F}\setminus {\cal T}_{\tau-1}$. Thus, $S$ has non-empty intersection with every set in ${\cal F}$, and is therefore a solution to $I$.

For the approximation ratio, note that the condition in Step \ref{step:checkLarge} is false when $i=\tau$, else the algorithm would not have reached Step \ref{step:existSmall}. Thus,
\[\begin{array}{lll}
\smallskip
|S| & = |S_{\tau}| + |U| - |U_{\tau-1}| &  [S_\tau\subseteq U_{\tau-1}\subseteq U]\\

\smallskip
& < |U|-(1-\frac{\mu}{d})|U_{\tau-1}| & [\mbox{The condition in Step \ref{step:checkLarge} is false}]\\

\smallskip
& \leq |U|- (1-\frac{\mu}{d})(\frac{1}{4})^{\frac{1}{\epsilon}d}|U| & [\mbox{Claim \ref{claim:universeLarge}}]\\

\smallskip
& = \left(1-(1-\frac{\mu}{d})(\frac{1}{4})^{\frac{1}{\epsilon}d}\right)|U|& \\

\smallskip
& \leq \left(1-(1-\frac{\mu}{d})(\frac{1}{4})^{\frac{1}{\epsilon}d}\right)d\cdot\mathsf{frac}(I) & [\mbox{Application of {\bf reduce}}] \\

& \leq \left(1-(1-\frac{\mu}{d})(\frac{1}{4})^{\frac{1}{\epsilon}d}\right)d\cdot\mathsf{opt}(I). & 
\end{array}\]
Hence, $\frac{|S|}{\mathsf{opt}(I)}\leq \left(1-(1-\frac{\mu}{d})(\frac{1}{4})^{\frac{1}{\epsilon}d}\right)d$. So, because $h(d,\epsilon)=\displaystyle{\coeffConstD(\frac{1}{4})^{\frac{d}{\epsilon}}}$, to conclude that $\frac{|S|}{\mathsf{opt}(I)}\leq  d(1-d\cdot h(d,\epsilon))$, it suffices to prove that $(1-\frac{\mu}{d})(\frac{1}{4})^{\frac{d}{\epsilon}}\geq \coeffConst(\frac{1}{4})^{\frac{d}{\epsilon}}$, which follows by substitution of $\mu=\frac{d+1}{2}$. This completes the proof of the claim.
\end{proof}

Lastly, we turn to conclude the proof of the theorem. First, because $\mathsf{frac}(I)\leq\mathsf{frac}(\widehat{I})$ (by the correctness of {\bf reduce}), Claim \ref{claim:callBound} implies that each call is of size as stated in the theorem. Further, this claim implies that with probability at least $9/10$, the algorithm does not exit in Step \ref{step:failure}. Under the assumption that the algorithm does not exit in Step \ref{step:failure}, notice that Claims~\ref{claim:exitLarge}  and~\ref{claim:exitSmall} ensure that $S$ is a solution and that $\frac{|S|}{\mathsf{opt}(I)}\leq \beta d(1-d\cdot h(d,\epsilon))$. So, by Lemma \ref{lem:dHSElementApprox} with $\rho=\displaystyle{\frac{d-1}{10}(\frac{1}{4})^{\frac{d}{\epsilon}}}$, at least one of the following conditions holds:
\begin{enumerate}
\item  $|\widehat{S}|-|S|\leq \rho\cdot\mathsf{opt}(\widehat{I})$, and hence $|\widehat{S}|\leq |S|+\displaystyle{\frac{d-1}{10}(\frac{1}{4})^{\frac{d}{\epsilon}}\cdot\mathsf{opt}(\widehat{I})}$. Then,
\[\begin{array}{ll}
\displaystyle{\frac{|\widehat{S}|}{\mathsf{opt}(\widehat{I})}}& \leq \displaystyle{\frac{|S|+\frac{d-1}{10}(\frac{1}{4})^{\frac{d}{\epsilon}}\cdot\mathsf{opt}(\widehat{I})}{\mathsf{opt}(\widehat{I})}}\\

&\leq \displaystyle{\frac{|S|}{\mathsf{opt}(I)} + \frac{d-1}{10}(\frac{1}{4})^{\frac{d}{\epsilon}}}\\

&\leq\displaystyle{\beta d(1-d\cdot h(d,\epsilon))+\frac{d-1}{10}(\frac{1}{4})^{\frac{d}{\epsilon}}}\\

&=\displaystyle{\beta d(1-d\cdot h(d,\epsilon))+d(d-1)\cdot h(d,\epsilon)}\\
&\leq \beta d(1-h(d,\epsilon)).
\end{array}\]

\item $\displaystyle{\frac{|\widehat{S}|}{\mathsf{opt}(\widehat{I})}\leq d-\frac{\rho}{d-1}=d-\frac{1}{10}(\frac{1}{4})^{\frac{d}{\epsilon}}=d(1-h(d,\epsilon))\leq \beta d(1-h(d,\epsilon))}$. 
\end{enumerate}
So, in both cases we got that
$\frac{|\widehat{S}|}{\mathsf{opt}(\widehat{I})}\leq\beta d(1-h(d,\epsilon)).$
This completes the proof.
\end{proof}

\begin{corollary}
For any fixed $\epsilon>0$, the {\sc $d$-Hitting Set} problem, parameterized by the optimum, admits a pure $d(1-h(d,\epsilon))$-approximate (randomized) kernelization protocol with call size $d\cdot\mathsf{opt}+2^{\frac{d}{\epsilon}}(d\cdot\mathsf{opt})^{1+\epsilon}$ (where the number of sets is at most $2^{\frac{d}{\epsilon}}(d\cdot\mathsf{opt})^{1+\epsilon}$) where $h(d,\epsilon)=\displaystyle{\coeffConstD(\frac{1}{4})^{\frac{d}{\epsilon}}}$ is a fixed positive constant that depends only on $d,\epsilon$.
\end{corollary}


\subsection{Relation Between a $(1+\epsilon)$-Approximate Kernelization Protocol for {\sc Vertex Cover} and the Ruzsa-Szemer\'{e}di Problem}\label{sec:VC}

We first present the following simple lemma.

\begin{lemma}\label{lem:charRuzsa}
Let $G$ be an $n$-vertex graph. Let $r=r(n),t=t(n)\in\mathbb{N}$. Let $U_1,U_2,\ldots,U_t\subseteq V(G)$ such that
\begin{itemize}
\item for all $i\in \{1,2,\ldots,t\}$, $G[U_i]$ has a matching $M_i$ of size at least $r$, and
\item for all distinct $i,j\in \{1,2,\ldots,t\}$, $E(G[U_i])\cap E(G[U_j])=\emptyset$.
\end{itemize}
Then, $G$ is a supergraph of an $(r,t)$-Ruzsa-Szemer\'{e}di graph.
\end{lemma}

\begin{proof}
For all $i\in \{1,2,\ldots,t\}$, let $M'_i$ be a matching in $G[U_i]$ of size exactly $r$, and let $U_i'\subseteq U_i$ be the vertices incident to at least one edge in $M'_i$. 
Let $G'$ be the graph on vertex set $\bigcup_{i=1}^tU_i'$ and edge set $\bigcup_{i=1}^tM'_i$. Notice that $M_1',M_2'\ldots,M_t'$ are matchings in $G'$.
Because for all distinct $i,j\in \{1,2,\ldots,t\}$, $E(G[U_i])\cap E(G[U_j])=\emptyset$, we have that $M'_1,M_2',\ldots,M_t'$ are pairwise disjoint, and hence form a partition of $E(G')$. Lastly, we claim that for all $i\in\{1,2,\ldots,t\}$, $M_i'$ is an induced matching in $G'$. Targeting a contradiction, suppose that this is false for some $i\in\{1,2,\ldots,t\}$. So, there exist  $u,v\in U_i'$ such that $\{u,v\}\notin M_i'$ but $\{u,v\}\in M_j'$ for some $j\in\{1,2,\ldots,t\}\setminus\{i\}$. However, this means that $\{u,v\}\in E(G[U_i])\cap E(G[U_j])$, which is a contradiction. This completes the proof.
\end{proof}

We are now ready to present our main theorem, which follows the lines of the kernelization protocol presented in Section \ref{sec:warmUp}. Clearly, this result makes sense only for choices of $c<\frac{1}{4}$ (so that the approximation ratio will be below $2$) and when $t=o(\sqrt{n})$, preferably $t=\OO(n^{\frac{1}{2}-\lambda})$ for $\lambda$ as close to $1/2$ as possible, so that the volume will be $\OO(\mathsf{opt}^{2-\lambda})$. Further, if $t$ is ``sufficiently small'' (depending on the desired number of rounds) whenever $c$ is a fixed constant, this yields a $(1+\epsilon)$-approximate kernelization protocol.

\begin{theorem}\label{thm:rusza}
Let $0<c<\frac{1}{4}$ be a fixed constant. For $r=r(n)=cn$, let $t=t(n)=\gamma(r)$.\footnote{That is, $t$ is the maximum value (as a function of $n$) such that there exists a $(r,t)$-Ruzsa-Szemer\'{e}di graph where $r=cn$ (see Definition \ref{def:ruzsa} and the discussion below it).} Then, the {\sc Vertex Cover} problem, parameterized by the  optimum, admits a $(1+4c)$-approximate (randomized) kernelization protocol with $t+1$ rounds and call size $2\mathsf{frac}+2(t+1)(2\mathsf{frac})^{1.5}$ (where the number of edges is at most $2(t+1)(2\mathsf{frac}^{1.5}$).
\end{theorem}

\begin{proof}
We first describe the algorithm. To this end, consider some input $\widehat{I}=\widehat{G}$. Then:
\begin{enumerate}
\item Call the {\bf reduce} procedure of the algorithm in Theorem \ref{thm:HSElementFrac} on $I$ to obtain a new instance $I=G$ where $|V(G)|\leq 2\mathsf{frac}(I)$.

\item Initialize $E_0=\emptyset$.

\item For $i=1,2,\ldots,t+1$:
	\begin{enumerate}
\item Sample $W_i$ from $E(G)$ as follows: Insert each edge $e\in E(G)$ to $W_i$ independently at random with probability $p=\displaystyle{\frac{1}{(2\mathsf{frac}(I))^{0.5}}}$.

\item If $|W_i|>2p|E(G)|$, then let $S$ be an arbitrary solution to $I$, and proceed directly to Step \ref{step:leftBet}. [\#Failure]

\item Call the oracle on $G_i=G-E(G)\setminus(E_{i-1}\cup W_i)$, and let $S_i$ denote its output. 

\item Let $M_i$ be some maximal matching in $G-S_i$, and let $T_i=E(G-S_i)$.

\item If $|M_i|<c|V(G)|$, then let $S=S_i\cup (\bigcup M_i)$,\footnote{That is, $S$ is the set that contains every vertex in $S_i$ as well every vertex incident to an edge in $M_i$.} and proceed directly to Step \ref{step:leftBet}. [\#Success]

\item If $|T_i|>2(2\mathsf{frac}(I))^{1.5}$, then let $S$ be an arbitrary solution to $I$, and proceed directly to Step \ref{step:leftBet}. [\#Failure]
	
\item Let $E_i=E_{i-1}\cup T_i$.	
	\end{enumerate}

\item\label{step:reach} Let $S$ be an arbitrary solution to $I$, and proceed directly to Step \ref{step:lift}. [\#Never Reach]

\item\label{step:leftBet} Call the {\bf lift} procedure of the algorithm in Theorem \ref{thm:HSElementFrac} on $\widehat{I},I,S$ to obtain a solution $\widehat{S}$ to $\widehat{I}$. Output $\widehat{S}$.
\end{enumerate}

Clearly, the algorithm runs in polynomial time, and only $t+1$ oracle calls are performed. Further, when we call the oracle on $G_i$,  then $|E(G_i)|\leq i\cdot 2p|E(G)|\leq 2(t+1)(2\mathsf{frac}(I))^{1.5}$ (due to {\bf reduce}). Thus, each oracle call is performed on an instance with at most $2\mathsf{frac}(I)$ vertices (as $|V(G)|\leq 2\mathsf{frac}(I)$ due to {\bf reduce}) and $2(t+1)(2\mathsf{frac}(I))^{1.5}$ edges, and since $\mathsf{frac}(I)\leq\mathsf{frac}(\widehat{I})$, the statement in the lemma regarding the call size is satisfied. 

Now, due to the correctness of {\bf lift}, it remains to show that we compute a solution $S$ to $I$ that, with probability at least $9/10$, is a $\beta(1+4c)$-approximate solution to $I$, where $\beta$ is the approximation ratio of the solutions returned by the oracle. Notice that if $S$ is computed in the step marked ``success'', say, at some iteration $i$, then clearly $|S|=|S_i|+2|M_i|<\beta\mathsf{opt}(I)+2c|V(G)|\leq \beta\mathsf{opt}(I)+4c\mathsf{opt}(I)\leq \beta(1+4c)\mathsf{opt}(I)$. Moreover, since $M_i$ is a maximal matching, every edge in $G$ that is not incident to $S_i$ must share an endpoint with at least one edge in $M_i$. So, $S$ is then a solution to $I$. Thus, it suffices to show that with probability at least $9/10$, $S$ is computed in the step marked by ``success''.

Just like in the proof of Theorem \ref{thm:warmUp}, we can show that the probability that the conditions in the steps marked by ``failure'' are not satisfied with probability at least $9/10$. So, it remains to show that we never reach Step \ref{step:reach}. Targeting a contradiction, suppose that we reach this step. For all $i\in\{1,2,\ldots,t+1\}$, let $U_i=V(G)\setminus S_i$. Then, $G[U_i]$ has a matching of size at least $r=c|V(G)|$ (that is $M_i$). Further, for all $1\leq i<j\leq t+1$, because $S_j$ is a vertex cover of $G_j$ and $E(G[U_i])=T_i\subseteq E(G_j)$, $E(G[U_i])\cap E(G[U_j])=\emptyset$. By Lemma \ref{lem:charRuzsa}, this means that $G$ is a supergraph of an $(r,t+1)$-Ruzsa-Szemer\'{e}di graph. However, this contradicts the definition of $t$. Thus, the proof is complete.
\end{proof}

\begin{corollary}
Let $0<c<\frac{1}{4}$ be a fixed constant. For $r=r(n)=cn$, let $t=t(n)=\gamma(r)$. Then, the {\sc Vertex Cover} problem, parameterized by the  optimum, admits a $(1+4c)$-approximate (randomized) kernelization protocol with $t+1$ rounds and call size $2\mathsf{opt}+2(t+1)(2\mathsf{opt})^{1.5}$ (where the number of edges is at most $2(t+1)(2\mathsf{opt})^{1.5}$).
\end{corollary}

\section{$(1+\epsilon)$-Approximate Linear-Vertex Kernels for Implicit {\sc $3$-Hitting Set} Problems}
\label{sec:lossyImplicit}
In this section, we present lossy kernels for two well-known implicit {\sc $3$-HS} problems, called {\sc Cluster Vertex Deletion}, and {\sc Feedback Vertex Set in Tournaments}. In {\sc Cluster Vertex Deletion}, given a graph $G$, the task is to compute a minimum-sized subset $S\subseteq V(G)$ such that $G-S$ is a cluster graph. In {\sc Feedback Vertex Set in Tournaments}, give a tournament $G$, the task is to compute a minimum-sized subset $S\subseteq V(G)$ such that $G-S$ is acyclic.
We attain a linear number of vertices at an approximation cost of only $(1+\epsilon)$ rather than $2$ as is given for {\sc $3$-HS} in Section \ref{sec:HSElement}.  Notably, both our algorithms follow similar lines, and we believe that the approach underlying their common parts may be useful when dealing also with other hitting and packing problems of constant-sized objects. In particular, in both algorithms we first ``reveal modules'' using essentially the same type of marking scheme, which yields a lossy rule, and afterwards we shrink the size of these modules using yet another rule that, unlike the first one, is problem-specific.




\subsection{{\sc Cluster Vertex Deletion}}\label{sec:P3}

Our lossy kernel will use Theorem \ref{thm:support} and consist of two rules, one lossy rule and one exact rule, each to be applied only once. The first rule (to which we will refer as the ``module revealing operation'') will ensure that all unmarked vertices in a clique (in some subgraph of the original graph, obtained by the removal of an approximate solution) form a module and furthermore that certain vertices among those removed have neighbors in only one of them, and the second one (``module shrinkage operation'') will reduce the size of each such module. For simplicity, we will actually merge them together to a single rule. We begin by reminding that  {\sc Cluster Vertex Deletion} can be interpreted as a special case of {\sc $3$-Hitting Set}:

\begin{proposition}[\cite{pcbook}]
A graph $G$ is a cluster graph if and only if it does not have any induced~$P_3$.
\end{proposition}

\begin{definition}\label{def:CVDtoHS}
Given a graph $G$, define the {\em {\sc $3$-Hitting Set} instance corresponding to $G$} by $\mathsf{HS}(G)=(V(G),\{\{u,v,w\}\subseteq V(G): G[\{u,v,w\}]$ is an induced $P_3\})$.
\end{definition}

\begin{corollary}
Let $G$ be a graph. Then, a subset $S\subseteq V(G)$ is a solution to the {\sc $3$-Hitting Set} instance corresponding to $G$ if and only if $G-S$ is a cluster graph.
\end{corollary}

To perform the module revealing operation, given a graph $G$, we will be working with an optimal solution $\alpha$ to the classic LP of the {\sc $3$-Hitting Set} instance corresponding to $G$. The approximate solution we will be working with will be the support of $\alpha$. For the sake of clarity, we slightly abuse notation and use vertices to refer both to vertices and to the variables corresponding to them, as well as use an instance of {\sc Cluster Vertex Deletion} to refer also to  the {\sc $3$-Hitting Set} instance corresponding to it when no confusion arises. We first show that the cliques in $G-\mathsf{support}(\alpha)$ are already modules in $G-\alpha^{-1}(1)$ (i.e., in the graph obtained by removing all vertices to which $\alpha$ assigns $1$). Thus, to reveal modules, we will only deal with vertices in $\alpha^{-1}(1)$.

\begin{lemma}\label{lem:modWRTnot1}
Let $G$ be a graph, and let $\alpha$ be a solution to the {\sc $3$-Hitting Set} instance corresponding to $G$. Let $C$ be a clique in $G-\mathsf{support}(\alpha)$. Then, $V(C)$ is a module in $G-\alpha^{-1}(1)$.
\end{lemma}

\begin{proof}
First, notice that as $\alpha$ is optimal, it does not assign values greater than $1$. Targeting a contradiction, suppose that $V(C)$ is not a module in $G-\alpha^{-1}(1)$. So, there exist vertices $v\in \mathsf{support}(\alpha)\setminus\alpha^{-1}(1)$ and $u,w\in V(C)$ such that $\{u,v\}\in E(G)$ and $\{w,v\}\notin E(G)$. Then,  as $\{u,w\}\in E(G)$ (since $C$ is a clique), $G[\{v,u,w\}]$ is an induced $P_3$. However, $\alpha$ assigns $0$ to the variables of $u$ and $w$, and a value smaller than $1$ to the variable of $v$, while the sum of these variables should be at least $1$ (because $\alpha$ is a solution). Thus, we have reached a contradiction.
\end{proof}

To deal with the vertices in $\alpha^{-1}(1)$, we now define a marking procedure that will be used by the first (implicit) rule.

\begin{definition}\label{def:markCVD}
Given $0<\epsilon<1$, a graph $G$ and an optimal solution $\alpha$ to the {\sc $3$-Hitting Set} instance corresponding to $G$, {\sf Marking}$(\epsilon,G,\alpha)$ is defined as follows.
\begin{enumerate}
\item For every vertex $v\in \alpha^{-1}(1)$, initialize {\sf mark}$(v)=\emptyset$.
\item For every vertex $v\in\alpha^{-1}(1)$:
	\begin{enumerate}
	\item Let $H_v$ be the graph defined as follows: $V(H_v)=V(G)\setminus(\mathsf{support}(\alpha)\cup(\bigcup_{u\in\alpha^{-1}(1)}\mathsf{mark}(u)))$, and $E(H_v)=\{\{w,r\}\subseteq V(H_v): G[\{v,w,r\}]$ is an induced $P_3\}$.
	\item Compute a maximal matching $\mu_v$ in $H_v$.\footnote{For example, by greedily picking edges so that the collection of edges remains a matching as long as it possible.}
	\item If $|\mu_v|>\frac{1}{\epsilon}$, then let $\nu_v$ be some (arbitrary) subset of $\mu_v$ of size exactly $\frac{1}{\epsilon}$, and otherwise let $\nu_v=\mu_v$.  Let $\mathsf{mark}(v)=\bigcup\nu_v$ (i.e., $\mathsf{mark}(v)$ is the set of vertices incident to edges in $\nu_v$).
	\end{enumerate}
\item For every vertex $v\in \alpha^{-1}(1)$, output $\mathsf{mark}(v)$. Moreover, output $D=\{v\in \alpha^{-1}(1): |\mathsf{mark}(v)|=\frac{1}{\epsilon}\}$.
\end{enumerate}
\end{definition}

We now prove that when all marked vertices are removed, the remainders of the cliques form modules in $G-D$.

\begin{lemma}\label{lem:createModules}
Given $0<\epsilon<1$, a graph $G$ and an optimal solution $\alpha$ to the {\sc $3$-Hitting Set} instance corresponding to $G$, let $\{\mathsf{mark}(v)\}|_{v\in\alpha^{-1}(1)},D$ be the output of {\sf Marking}$(\epsilon,G,\alpha)$. Then, the vertex set of every clique $C$ in $G-(\mathsf{support}(\alpha)\cup(\bigcup_{v\in\alpha^{-1}(1)}\mathsf{mark}(v)))$ is a module in $G-D$.
\end{lemma}

\begin{proof}
Consider some clique $C$ in $G-(\mathsf{support}(\alpha)\cup(\bigcup_{v\in\alpha^{-1}(1)}\mathsf{mark}(v)))$. Clearly, every vertex in $G-\mathsf{support}(\alpha)$ is adjacent to either all vertices in $C$ (when they belong to the same clique in $G-\mathsf{support}(\alpha)$) or to none (when they belong to different cliques). Further, due to Lemma \ref{lem:modWRTnot1}, every vertex in $\mathsf{support}(\alpha)\setminus\alpha^{-1}(1)$ also has this property. So, it remains to prove that every vertex in $\alpha^{-1}(1)\setminus D$ also has this property. To this end, consider some vertex $v\in\alpha^{-1}(1)\setminus D$. Targeting a contradiction, suppose that there exist vertices $u,w\in V(C)$ such that $\{u,v\}\in E(G)$ but $\{w,v\}\notin E(G)$. As $V(C)\cap (\bigcup_{v'\in\alpha^{-1}(1)}\mathsf{mark}(v')))=\emptyset$, $H_v$ (in Definition \ref{def:markCVD}) contained the edge $\{u,w\}$. Moreover, neither $u$ nor $w$ was inserted into $\mathsf{mark}(v)$ and hence, as $\nu_v=\mu_v$ (because $v\notin D$), none of them is incident to an edge in $\mu_v$. However, this contradicts that $\mu_v$ is a maximal matching, as we can insert $\{u,w\}$ to it and it would remain a matching.
\end{proof}

We now argue that every optimal solution contains all of the vertices of $D$ except of an $\epsilon$-fraction of the optimum, and hence it is not ``costly'' to seek only solutions that contain $D$.

\begin{lemma}\label{lem:cvdSafe1}
Let $I=G$ be an instance of {\sc Cluster Vertex Deletion}.
Given $0<\epsilon<1$, $G$ and an optimal solution $\alpha$ to the {\sc $3$-Hitting Set} instance corresponding to $G$, let $\{\mathsf{mark}(v)\}|_{v\in\alpha^{-1}(1)},D$ be the output of {\sf Marking}$(\epsilon,G,\alpha)$. Let $S^\star$ be an optimal solution to $I$. Then, $|D\setminus S^\star|\leq\epsilon\mathsf{opt}(I)$.
\end{lemma}

\begin{proof}
Consider some vertex $v\in D$. Notice that $v$ together with any edge in $\nu_v$ form an induced $P_3$ in $G$. Thus, if $v\notin S^\star$, then from every edge in $\nu_v$, at least one vertex must belong to $S^\star$. As $\nu_v$ is a matching, and its size is $\frac{1}{\epsilon}$, this means that $S^\star$ has to contain at least $\frac{1}{\epsilon}$ vertices from $\mathsf{mark}(v)$. 
As the sets assigned by $\mathsf{mark}$ are pairwise disjoint, we have that $|D\setminus S^\star|$ can be of size at most $\epsilon|S^\star|=\epsilon\mathsf{opt}(I)$.
\end{proof}

Intuitively, the arguments above naturally give rise to a reduction rule that deletes $D$. However, a minor technicality arises---that is, we will need to transmit $\alpha$ and the marked sets to the reduced instance in order for our next arguments to work, which, when complying with necessary formalities, requires to define an annotated version of the problem. We avoid this by merging the rule implicitly in our main rule later, which simplifies the presentation.

Before we proceed to shrink the size of the modules, we argue that every vertex outside them (except for those in $D$) has neighbors in at most one of them.

\begin{lemma}\label{lem:cvdSeeOne}
Given $0<\epsilon<1$, a graph $G$ and an optimal solution $\alpha$ to the {\sc $3$-Hitting Set} instance corresponding to $G$, let $\{\mathsf{mark}(v)\}|_{v\in\alpha^{-1}(1)},D$ be the output of {\sf Marking}$(\epsilon,G,\alpha)$. Then, for every vertex $v\in \mathsf{support}(\alpha)\setminus D$, $N_G(v)\setminus(\mathsf{support}(\alpha)\cup(\bigcup_{v\in\alpha^{-1}(1)}\mathsf{mark}(v)))$ is either empty or equals the vertex set of exactly one clique $C$ in $G-(\mathsf{support}(\alpha)\cup(\bigcup_{v\in\alpha^{-1}(1)}\mathsf{mark}(v)))$.
\end{lemma}

\begin{proof}
Consider some vertex $v\in \mathsf{support}(\alpha)\setminus D$.  Targeting a contradiction, suppose that the lemma is false with respect to $v$. Due to Lemma \ref{lem:createModules}, this necessarily means that there exist two distinct cliques $C,C'$ in $G-(\mathsf{support}(\alpha)\cup(\bigcup_{v\in\alpha^{-1}(1)}\mathsf{mark}(v)))$ such that $v$ has neighbors in both. So, let $u\in V(C)$ and $w\in V(C')$ be such that $\{u,v\},\{w,v\}\in E(G)$. Observe that as $u,w$ belong to different cliques, $\{u,w\}\notin E(G)$. As $(V(C)\cup V(C'))\cap (\bigcup_{v'\in\alpha^{-1}(1)}\mathsf{mark}(v')))=\emptyset$, $H_v$ (in Definition \ref{def:markCVD}) contained the edge $\{u,w\}$. Moreover, neither $u$ nor $w$ was inserted into $\mathsf{mark}(v)$ and hence, as $\nu_v=\mu_v$ (because $v\notin D$), none of them is incident to an edge in $\mu_v$. However, this contradicts that $\mu_v$ is a maximal matching, as we can insert $\{u,w\}$ to it and it would remain a matching.
\end{proof}

We now proceed to shrink the size of the modules we have just revealed. We note that this part is, in fact, exact (i.e.~$1$-approximate). For this purpose, we start with the following observation.

\begin{lemma}\label{lem:moduleUnitHelper}
Let $I=G$ be an instance of {\sc Cluster Vertex Deletion}.
Let $T\subseteq V(G)$ be a module in $G$ such that $G[T]$ is a clique. Then, any induced $P_3$ in $G$ that contains at least one vertex from $T$, contains exactly one vertex from $T$ and at least one vertex from $N_G(T)$.
\end{lemma}

\begin{proof}
Let $P$ be an induced $P_3$ in $G$ that contains at least one vertex from $T$. Clearly, $P$ cannot contain three from $T$ as $T$ induces a clique. Further, if it contains exactly two vertices from $T$, then they must be adjacent in $P$ as $T$ induces a clique, but then the third vertex in $P$ will be a neighbor of one but non-neighbor of the other, which is a contradiction because $T$ is a module. So, $P$ contains exactly one vertex from $T$, which also implies that it must contain at least one vertex from $N_G(T)$.
\end{proof}

We use the above lemma to give a simple lemma that says that if at least one vertex of a module that induces a clique is deleted, then unless that deletion is unnecessary, all of its vertices are deleted.

\begin{lemma}\label{lem:moduleUnit}
Let $I=G$ be an instance of {\sc Cluster Vertex Deletion}.
Let $T\subseteq V(G)$ be a module in $G$ such that $G[T]$ is a clique. Then, for every solution $S$ to $I$, either $T\subseteq S$ or $S\setminus T$ is also a solution to $I$.
\end{lemma}

\begin{proof}
Consider a solution $S$ to $I$ such that $T$ is not contained in $S$. So, there exists some vertex $v\in T\setminus S$. We claim that $S\setminus T$ is also a solution to $I$, which will complete the proof. Targeting a contradiction, suppose that $S\setminus T$ is not a solution to $I$, thus there exists an induced $P_3$, say, $P$, in $G-(S\setminus T)$.  Observe that $P$ contains at least one vertex from $T$, as $S$ is a solution. By Lemma \ref{lem:moduleUnitHelper}, this means that it must contain exactly one vertex from $T$. So, let $u$ denote the only vertex in $T\cap V(P)$. As $T$ is a module, by replacing $u$ by $v$, we obtain yet another induced $P_3$, and this one belongs to $G-S$. This is a contradiction as $S$ is a solution to $I$.
\end{proof}

Further, given a solution that contains $D$, we may exchange a clique by its neighborhood  and still have a solution, as stated below.

\begin{lemma}\label{lem:moduleExchange}
Let $I=G$ be an instance of {\sc Cluster Vertex Deletion}.
Given $0<\epsilon<1$, $G$ and an optimal solution $\alpha$ to the {\sc $3$-Hitting Set} instance corresponding to $G$, let $\{\mathsf{mark}(v)\}|_{v\in\alpha^{-1}(1)},D$ be the output of {\sf Marking}$(\epsilon,G,\alpha)$. Let $S^\star$ be a solution to $I$. Then, for every clique $C$ in $G-(\mathsf{support}(\alpha)\cup(\bigcup_{v\in\alpha^{-1}(1)}\mathsf{mark}(v)))$, we have that $(S^\star\setminus V(C))\cup N_G(V(C))$ is a solution to~$I$.
\end{lemma}

\begin{proof}
Consider some clique $C$ in $G-(\mathsf{support}(\alpha)\cup(\bigcup_{v\in\alpha^{-1}(1)}\mathsf{mark}(v)))$. As $S^\star$ is a solution to $I$, every induced $P_3$ in $G-(S^\star\setminus V(C))$ must contain at least one vertex from $C$. However, in $G-N_G(V(C))$, $C$ is an isolated clique and hence there exists no induced $P_3$ that contains at least one of its vertices. So, $G-(S^\star\setminus V(C))\cup N_G(V(C))$ does not have any induced $P_3$, and hence $(S^\star\setminus V(C))\cup N_G(V(C))$ is a solution to $I$.
\end{proof}

So, Lemmata \ref{lem:moduleUnit} and \ref{lem:moduleExchange} naturally give rise to a reduction operation where each clique whose size is larger than the size of its neighbrhood is shrunk to the size of its neighborhood.

Finally, we devise our merged reduction rule:

\begin{definition}\label{def:CVDRule}
The {\em {\sc  Cluster Vertex Deletion} reduction rule} is defined as follows:
\begin{itemize}
\item {\bf reduce:} Let $I=G$ be an instance of {\sc Cluster Vertex Deletion}. Use the algorithm in Proposition \ref{prop:solveLP} to compute an optimal solution $\alpha$ to the classic LP corresponding to it (Definitions \ref{def:HSLP} and \ref{def:CVDtoHS}). Let $\{\mathsf{mark}(v)\}|_{v\in\alpha^{-1}(1)},D$ be the output of {\sf Marking}$(\epsilon,G,\alpha)$.
Output $I'=G'$ where $G'$ is obtained from $G-D$ as follows: for every clique $C$ in $G-(\mathsf{support}(\alpha)\cup(\bigcup_{v\in\alpha^{-1}(1)}\mathsf{mark}(v)))$, delete (arbitrarily chosen) $\max\{|V(C)|-|N_{G-D}(V(C))|\}$ vertices from $C$.

\item {\bf lift:} Given $I,I'$ and a solution $S'$ to $I'$, output  $S=D\cup S''$ where $S''$ is obtained from $S'$ as follows: for every clique $C$ in $G-(\mathsf{support}(\alpha)\cup(\bigcup_{v\in\alpha^{-1}(1)}\mathsf{mark}(v)))$ such that $V(C)$ is not contained in $G'$ and $V(C)\subseteq S'$, remove $V(C)\cap S'$ and add $N_G(V(C))$ instead.
\end{itemize}
\end{definition}

Before we prove our main theorem, we present a simple lemma that will help us derive a tighter bound on the number of vertices in the output graph.

\begin{lemma}\label{lem:cvdSupport}
Let $I=G$ be an instance of {\sc Cluster Vertex Deletion}, and let $\alpha$ be a solution to the {\sc $3$-Hitting Set} instance corresponding to $G$. Then, $|\mathsf{support}(\alpha)|\leq 3\mathsf{frac}(I)-2|\alpha^{-1}(1)|$.
\end{lemma}

\begin{proof}
Let $\beta$ denote the restriction of $\alpha$ to $G-\alpha^{-1}(1)$. Observe that $\beta$ is a solution to the classic LP of the {\sc $3$-Hitting Set} instance corresponding to $G$ and its value is  $\mathsf{frac}(I)-|\alpha^{-1}(1)|$. Further, consider some solution $\beta'$ to the classic LP of the {\sc $3$-Hitting Set} instance corresponding to $G'=G-\alpha^{-1}(1)$, thought of as an instance $I'$ of {\sc Cluster Vertex Deletion}. Then, by extending $\beta'$ to assign $1$ to each vertex in $\alpha^{-1}(1)$, we obtain a solution $\alpha'$ to the classic LP of the {\sc $3$-Hitting Set} instance corresponding to $G$. As $\alpha$ is optimal, its value is at most that of $\alpha'$. Thus, it must hold that the value of $\beta'$ is at least that of $\beta$. Since the choice of $\beta'$ was arbitrary, this implies that $\beta$ is optimal. Hence, by Theorem \ref{thm:support}, 
\[|\mathsf{support}(\beta)|\leq 3\cdot\mathsf{frac}(I')=3\cdot(\mathsf{frac}(I)-|\alpha^{-1}(1)|).\]
Thus, we have that
\[|\mathsf{support}(\alpha)|=|\mathsf{support}(\beta)|+|\alpha^{-1}(1)|\leq 3\mathsf{frac}(I)-2|\alpha^{-1}(1)|.\]
This completes the proof.
\end{proof} 

Based on Lemmata \ref{lem:createModules}, \ref{lem:cvdSafe1}, \ref{lem:cvdSeeOne}, \ref{lem:moduleUnit}, \ref{lem:moduleExchange} and \ref{lem:cvdSupport}, we are now ready to prove the main theorem of this subsection.

\begin{theorem}\label{thm:cvd}
Let $0<\epsilon<1$. The {\sc Cluster Vertex Deletion} problem, parameterized by the fractional optimum of the classic LP, admits a $(1+\epsilon)$-approximate $\max(6,\frac{4}{\epsilon})\cdot\mathsf{frac}$-vertex kernel.
\end{theorem}

\begin{proof}
Our lossy kernelization algorithm consists only of the {\sc  Cluster Vertex Deletion} reduction rule. Clearly, it runs in polynomial time. 

First, we consider the number of vertices in the output graph $G'$ of {\bf reduce}. By Lemma~\ref{lem:cvdSupport}, $|\mathsf{support}(\alpha)|\leq 3\cdot\mathsf{frac}(I)-2|\alpha^{-1}(1)|$ (I). Moreover, $|\bigcup_{v\in\alpha^{-1}(1)}\mathsf{mark}(v)| = 2|\bigcup_{v\in\alpha^{-1}(1)}\nu_v|\leq \frac{2}{\epsilon}|\alpha^{-1}(1)|$ (II). By the definition of the reduction rule, for every clique $C$ in $G'-(\mathsf{support}(\alpha)\cup(\bigcup_{v\in\alpha^{-1}(1)}\mathsf{mark}(v)))$, $|V(C)|\leq |N_{G'}(V(C))|$. Additionally, by Lemma \ref{lem:cvdSeeOne}, the neigborhood sets of these cliques are pairwise vertex disjoint. This implies that, altogether, these cliques contain at most $|\mathsf{support}(\alpha)\setminus\alpha^{-1}(1)\cup(\bigcup_{v\in\alpha^{-1}(1)}\mathsf{mark}(v)))|$ vertices. Thus, because the vertex set of $G'$ consists only of these cliques and of $\mathsf{support}(\alpha)\setminus\alpha^{-1}(1)\cup(\bigcup_{v\in\alpha^{-1}(1)}\mathsf{mark}(v))$, we conclude that
\[\begin{array}{lll}
|V(G')| & \leq 2|\mathsf{support}(\alpha)| -2|\alpha^{-1}(1)| + 2|\bigcup_{v\in\alpha^{-1}(1)}\mathsf{mark}(v)| & [\mbox{Last two sentences}]\\

&\leq 2|\mathsf{support}(\alpha)|-2|\alpha^{-1}(1)| + \frac{4}{\epsilon}|\alpha^{-1}(1)|& [\mbox{(II)}]\\

&\leq 6\cdot\mathsf{frac}(I)-4|\alpha^{-1}(1)|-2|\alpha^{-1}(1)| + \frac{4}{\epsilon}|\alpha^{-1}(1)|& [\mbox{(I)}]\\

& = 6\cdot\mathsf{frac}(I) + (\frac{4}{\epsilon}-6)|\alpha^{-1}(1)| &\\

& \leq \max(6,\frac{4}{\epsilon})\cdot\mathsf{frac}(I) & [|\alpha^{-1}(1)|\leq \mathsf{frac}(I)].
\end{array}\]

We turn to prove that {\bf lift} returns a solution having the desired approximation ratio. To this end, suppose that it is given $I,I',S'$ where $S'$ is a solution to $I'$. First, notice that $S'\cup (V(G)\setminus V(G'))$ is a solution to $I$. Thus, because $S$ can be obtained from $S'\cup (V(G)\setminus V(G'))$ by doing deletion and exchange operations as described in Lemmata \ref{lem:moduleUnit} and \ref{lem:moduleExchange}, these lemmata imply that $S$ is a solution to $I$. 

Now, we consider the approximation ratio of $S$. For this, on the one hand, let $\widehat{S}$ be a minimal solution contained in $S'$. By Lemmata \ref{lem:createModules} and \ref{lem:moduleUnit}, for every clique $C'$ in $G'-(\mathsf{support}(\alpha)\cup(\bigcup_{v\in\alpha^{-1}(1)}\mathsf{mark}(v)))$, either $V(C')\subseteq \widehat{S}$ or $V(C')\cap\widehat{S}=\emptyset$. So, because every clique $C$ in $G-(\mathsf{support}(\alpha)\cup(\bigcup_{v\in\alpha^{-1}(1)}\mathsf{mark}(v)))$ such that $V(C)$ is not contained in $G'$ satisfies that $|V(C)\cap V(G')|=|N_{G-D}(V(C))|$, we know that $|S\setminus D|=|\widehat{S}|\leq |S'|$.  Hence, {\em (i)} $|S|\leq |S'|+|D|$. On the other hand, let $S^\star$ be an optimal solution to $I$. By Lemma \ref{lem:cvdSafe1}, $|D\setminus S^\star|\leq \epsilon\mathsf{opt}(I)$. Thus, $|S^\star\setminus D|=\mathsf{opt}(I)-|S^\star\cap D|\geq \mathsf{opt}(I)-(|D|-\epsilon\mathsf{opt}(I))=(1+\epsilon)\mathsf{opt}(I)-|D|$. Further, as $S^\star\cap V(G')$, which is a subset of $S^\star\setminus D$, is a solution to $I'$, we have that $\mathsf{opt}(I')\leq |S^\star\setminus D|$, and hence {\em (ii)} $\mathsf{opt}(I')\leq (1+\epsilon)\mathsf{opt}(I)-|D|$. From {\em (i)} and {\em (ii)}, we conclude that
\[\frac{|S|}{\mathsf{opt}(I)}\leq (1+\epsilon)\frac{|S'|+|D|}{\mathsf{opt}(I')+|D|}\leq (1+\epsilon)\max\{\frac{|S'|}{\mathsf{opt}(I')},\frac{|D|}{|D|}\} = (1+\epsilon)\frac{|S'|}{\mathsf{opt}(I')}.\]
Here, the last inequality follows from Proposition \ref{prop:numbers}. This completes the proof.
\end{proof}

\begin{corollary}
Let $0<\epsilon<1$. The {\sc Cluster Vertex Deletion} problem, parameterized by the optimum, admits a $(1+\epsilon)$-approximate $\max(6,\frac{4}{\epsilon})\cdot\mathsf{opt}$-vertex kernel.
\end{corollary}

Due to Lemma \ref{lem:lossyKerOptToK}, we also have the following corollary of Theorem \ref{thm:cvd}.

\begin{corollary}
Let $0<\epsilon<1$. The {\sc Cluster Vertex Deletion} problem, parameterized by a bound $k$ on the solution size, admits a $(1+\epsilon)$-approximate $\max(\frac{6}{1+\epsilon},\frac{4}{(1+\epsilon)\epsilon})\cdot(k+1)$-vertex kernel.
\end{corollary}

\subsection{{\sc Feedback Vertex Set in Tournaments}}\label{sec:fvst}

Our lossy kernel will use Theorem \ref{thm:support} and consist of two lossy rules, each to be applied only once. The first rule (to which we will refer as the ``module revealing operation'') will ensure that,
with respect to some linear order on the vertices not in some approximate solution, all consecutive unmarked vertices between two marked vertices form a module and furthermore that there is an essentially unique position to place each vertex (including those in the approximate solution) between them, and the second one (``module shrinkage operation'') will reduce the size of each such module. For simplicity, we will actually merge them together to a single rule.
We begin by reminding that  {\sc Feedback Vertex Set in Tournaments} can be interpreted as a special case of {\sc $3$-Hitting Set}:

\begin{proposition}[\cite{pcbook}]
A tournament $G$ is acyclic if and only if it does not have any triangle (i.e., a directed cycle on three vertices).
\end{proposition}

\begin{definition}\label{def:FVSTtoHS}
Given a tournament $G$, define the {\em {\sc $3$-Hitting Set} instance corresponding to $G$} by $\mathsf{HS}(G)=(V(G),\{\{u,v,w\}\subseteq V(G): G[\{u,v,w\}]$ is triangle $\})$.
\end{definition}

\begin{corollary}\label{cor:charByTriangles}
Let $G$ be a tournament. Then, a subset $S\subseteq V(G)$ is a solution to the {\sc $3$-Hitting Set} instance corresponding to $G$ if and only if $G-S$ is acyclic.
\end{corollary}

To perform the module revealing operation, given a graph $G$, we will be working with an optimal solution $\alpha$ to the classic LP of the {\sc $3$-Hitting Set} instance corresponding to $G$. The approximate solution we will be working with will be the support of $\alpha$. For the sake of clarity, we slightly abuse notation and use vertices to refer both to vertices and to the variables corresponding to them, as well as use an instance of {\sc Feedback Vertex Set in Tournaments} to refer also to  the {\sc $3$-Hitting Set} instance corresponding to it when no confusion arises. We will use the following well-known characterization of acyclic digraphs.

\begin{proposition}[Folklore]\label{prop:linOrder}
A digraph $G$ is acyclic if and only if there exists a linear order $<$ on $V(G)$ such that for every arc $(u,v)\in E(G)$, $u<v$. Moreover, given an acyclic digraph $G$, such an order is computable in linear time, and if $G$ is a tournament, then this order is unique.
\end{proposition}

This gives rise to the following definition.

\begin{definition}
Let $G$ be a tournament, and let $\alpha$ be a solution to the {\sc $3$-Hitting Set} instance corresponding to $G$. Then, the {\em linear order induced by $\alpha$}, denoted $<_{\alpha}$, is the unique linear ordering of $V(G)\setminus \mathsf{support}(\alpha)$ such that for every arc $(u,v)\in E(G-\mathsf{support}(\alpha))$, $u<_{\alpha}v$. We say that two vertices $u,v\in V(G-\mathsf{support}(\alpha))$ are {\em consecutive} in $<_\alpha$ if $u<_\alpha v$ and there is no vertex $w\in V(G-\mathsf{support}(\alpha))$ such that $u<_\alpha w<_\alpha v$; then, $u$ is called the {\em successor} of $v$, and $v$ is called the {\em predecessor} of $u$. 
\end{definition}

We further define the notion of a {\em position} based on this order.

\begin{definition}
Let $G$ be a tournament, and let $\alpha$ be a solution to the {\sc $3$-Hitting Set} instance corresponding to $G$. Let $M\subseteq V(G)\setminus\mathsf{support}(\alpha)$. Then, a vertex $v\in \mathsf{support}(\alpha)\cup M$ {\em $M$-fits} $<_\alpha$ if one of the following conditions holds.
\begin{itemize}
\item For all $u\in V(G)\setminus (\mathsf{support}(\alpha)\cup M)$, $(v,u)\in E(G)$. In this case, we say that $v$ has {\em $0$-position (with respect to $M$)}.
\item There exists $u\in V(G)\setminus (\mathsf{support}(\alpha)\cup M)$ such that for every $r\in V(G)\setminus (\mathsf{support}(\alpha)\cup M)$ where $r\leq_\alpha u$, $(r,v)\in E(G)$, and for every $r\in V(G)\setminus (\mathsf{support}(\alpha)\cup M)$ where $r>_\alpha u$, $(v,r)\in E(G)$. In this case, we say that $v$ has {\em $u$-position (with respect to $M$)}.
\end{itemize}
\end{definition}

We suppose that a $0$-position is the lowest possible, that is, $0 <_{\alpha} u$ for all $u\in V(G)\setminus(\mathsf{suppot}(\alpha)\cup M)$. The following observations are immediate.

\begin{observation}
Let $G$ be a tournament, and let $\alpha$ be a solution to the {\sc $3$-Hitting Set} instance corresponding to $G$. Let $M\subseteq V(G)\setminus\mathsf{support}(\alpha)$. Let $v\in \mathsf{support}(\alpha)\cup M$ be a vertex that {\em $M$-fits} $<_\alpha$. Then,  there exists exactly one element $u\in \{0\}\cup(V(G)\setminus (\mathsf{support}(\alpha)\cup M))$ such that $v$ has $u$-position.
\end{observation}

\begin{observation}\label{obs:Mfits}
Let $G$ be a tournament, and let $\alpha$ be a solution to the {\sc $3$-Hitting Set} instance corresponding to $G$. Let $M\subseteq V(G)\setminus\mathsf{support}(\alpha)$. Then, every vertex in $M$  {\em $M$-fits} $<_\alpha$. 
\end{observation}

We first show that the vertices in $\mathsf{support}(\alpha)\setminus\alpha^{-1}(1)$ already $\emptyset$-fit $<_\alpha$ (so, they also $M$-fit $<_{\alpha}$ with respect to any $M\subseteq V(G)\setminus\mathsf{support}(\alpha)$). Thus, to reveal modules that give rise to unique positions, we will only deal with vertices in $\alpha^{-1}(1)$.

\begin{lemma}\label{lem:modWRTnot1pack}
Let $G$ be a tournament, and let $\alpha$ be a solution to the {\sc $3$-Hitting Set} instance corresponding to $G$. Let $v\in\mathsf{support}(\alpha)\setminus\alpha^{-1}(1)$. Then, $v$ $\emptyset$-fits $<_\alpha$.
\end{lemma}

\begin{proof}
First, notice that as $\alpha$ is optimal, it does not assign values greater than $1$. Thus, $\alpha(x_v)<1$. So, $G-(\mathsf{support}(\alpha)\setminus\{v\})$ does not have a triangle, else the sum of the variables of its vertices will be less than $1$, contradicting that $\alpha$ is a solution. By Corollary \ref{cor:charByTriangles} and Proposition \ref{prop:linOrder}, this means that $G-(\mathsf{support}(\alpha)\setminus\{v\})$ admits a unique linear order $<$ such that for every arc $x,y)\in E(G-(\mathsf{support}(\alpha)\setminus\{v\}))$, $x<y$, and its restriction to $G-\mathsf{support}(\alpha)$ must equal $<_{\alpha}$. This directly implies the lemma, where if $v$ is first in $<$ then it has $0$-position, and otherwise it has $u$-position where $u$ is its predecessor in $<$.
\end{proof}

To deal with the vertices in $\alpha^{-1}(1)$, we define the following marking procedure.

\begin{definition}\label{def:markFVST}
Given $0<\delta<1$, a tournament $G$ and an optimal solution $\alpha$ to the {\sc $3$-Hitting Set} instance corresponding to $G$, {\sf Marking}$(\delta,G,\alpha)$ is defined as follows.
\begin{enumerate}
\item For every vertex $v\in \alpha^{-1}(1)$, initialize {\sf mark}$(v)=\emptyset$.
\item For every vertex $v\in\alpha^{-1}(1)$:
	\begin{enumerate}
	\item Define the graph $H_v$ as follows: $V(H_v)=V(G)\setminus(\mathsf{support}(\alpha)\cup(\bigcup_{u\in\alpha^{-1}(1)}\mathsf{mark}(u)))$, and $E(H_v)=\{\{w,r\}\subseteq V(H_v): G[\{v,w,r\}]$ is a triangle$\}$.
	\item Compute a maximal matching $\mu_v$ in $H_v$.
	\item If $|\mu_v|>\frac{1}{\delta}$, then let $\nu_v$ be some (arbitrary) subset of $\mu_v$ of size exactly $\frac{1}{\delta}$, and otherwise let $\nu_v=\mu_v$.  Let $\mathsf{mark}(v)=\bigcup\nu_v$ (i.e., $\mathsf{mark}(v)$ is the set of vertices incident to edges in $\nu_v$).
	\end{enumerate}
\item For every vertex $v\in \alpha^{-1}(1)$, output $\mathsf{mark}(v)$. Moreover, output $M=\bigcup\{\mathsf{mark}(v)\}|_{v\in\alpha^{-1}(1)},$ $D=\{v\in \alpha^{-1}(1): |\mathsf{mark}(v)|=\frac{1}{\epsilon}\}$.
\end{enumerate}
\end{definition}

We define {\em regions} based on marked vertices as follows. We will not need this definition for our proof, but we still give it since it provides some intuition regarding which modules are created. We remark that this is the only notion/argument in this subsection that is not necessary.

\begin{definition}
Given $0<\delta<1$, a tournament $G$ and an optimal solution $\alpha$ to the {\sc $3$-Hitting Set} instance corresponding to $G$, let $\{\mathsf{mark}(v)\}|_{v\in\alpha^{-1}(1)},M,D$ be the output of {\sf Marking}$(\delta,G,\alpha)$. Then, an {\em $(M,D)$-region} ({\em region} for short) is a maximal subset $U\subseteq V(G)\setminus(\mathsf{support}(\alpha)\cup M)$ such that there do not exist vertices $v\in M$, $u,w\in V(G)\setminus(\mathsf{support}(\alpha)\cup M)$ such that $u<_\alpha v<_\alpha w$. The collection of regions is denoted by ${\cal R}$.
\end{definition}

We prove that all vertices except for those in $D$, and not just those in $\mathsf{support}(\alpha)\setminus\alpha^{-1}(1)$, now have unique positions when marked vertices are removed.

\begin{lemma}\label{lem:createModulesPack}
Given $0<\delta<1$, a tournament $G$ and an optimal solution $\alpha$ to the {\sc $3$-Hitting Set} instance corresponding to $G$, let $\{\mathsf{mark}(v)\}|_{v\in\alpha^{-1}(1)},M,D$ be the output of {\sf Marking}$(\delta,G,\alpha)$. Then, every vertex $v\in(\mathsf{support}(\alpha)\setminus D)\cup M$ $M$-fits $<_\alpha$.
\end{lemma}

\begin{proof}
By Observation \ref{obs:Mfits}, the lemma is true for vertices in $M$. So, let $v\in\mathsf{support}(\alpha)\setminus D$. Due to Lemma \ref{lem:modWRTnot1pack}, the lemma is correct if $v\notin\alpha^{-1}(1)$, so we next suppose that $v\notin\alpha^{-1}(1)$. We claim that $G-((\mathsf{support}(\alpha)\cup M)\setminus\{v\})$ does not have a triangle. Targeting a contradiction, suppose that it has a triangle $T$. Then, as $\mathsf{support}(\alpha)$  is a solution, necessarily $v$ belongs to the triangle. So, denote $V(T)=\{v,u,w\}$. However, we have that $(u,w)\in\mu_v$ but $(u,w)\notin\nu_v$. This is a contradiction since $v\notin D$.
So far, we conclude  $G-((\mathsf{support}(\alpha)\cup M)\setminus\{v\})$ does not have a triangle. Thus, by Corollary \ref{cor:charByTriangles} and Proposition \ref{prop:linOrder}, this means that $G-((\mathsf{support}(\alpha)\cup M)\setminus\{v\})$ admits a unique linear order $<$ such that for every arc $(x,y)\in E(G-((\mathsf{support}(\alpha)\cup M)\setminus\{v\}))$, $x<y$, and its restriction to $G-\mathsf{support}(\alpha)$ must equal $<_{\alpha}$. This directly implies the lemma, where if $v$ is first in $<$ then it has $0$-position, and otherwise it has $u$-position where $u$ is its predecessor in $<$.
\end{proof}

We remark that Lemma \ref{lem:createModulesPack} will be implicitly used throughout, specifically when we consider vertices $v\in\mathsf{support}(\alpha)\setminus D$ and implicitly suppose that the definition of their position is valid.  An easy consequence of Lemma \ref{lem:createModulesPack} is that all regions are modules. However, we will not need to directly use this, but rather use  Lemma \ref{lem:createModulesPack}.
%
Moreover, as a consequence of Lemma \ref{lem:createModulesPack}, we can characterize the triangles in $G$ as follows.

\begin{lemma}\label{lem:traignleCharacterization}
Given $0<\delta<1$, a tournament $G$ and an optimal solution $\alpha$ to the {\sc $3$-Hitting Set} instance corresponding to $G$, let $\{\mathsf{mark}(v)\}|_{v\in\alpha^{-1}(1)},M,D$ be the output of {\sf Marking}$(\delta,G,\alpha)$. Then, every triangle in $G-D$ consists of either 
\begin{enumerate}
\item three vertices of $(\mathsf{support}(\alpha)\setminus D)\cup M$, or
\item a vertex $v\in\mathsf{support}(\alpha)\setminus D$, a vertex $u\in(\mathsf{support}(\alpha)\setminus D)\cup M$ and a vertex $w\in V(G)\setminus (\mathsf{support}(\alpha)\cup M)$ such that either {\em (i)} $(u,v)\in E(G)$, $v$ is of position $0$ or $r<_\alpha w$, and $u$ is of position $r'\geq_\alpha w$, or {\em (ii)} $(v,u)\in E(G)$, $u$ is of position $0$ or $r<_\alpha w$, and $v$ is of position $r'\geq_\alpha w$.
\end{enumerate}
\end{lemma}

\begin{proof}
Let $T$ be a triangle in $G-D$. Because $\mathsf{support}(\alpha)$ is a solution to $G$, $T$ must contain at least one vertex from $\mathsf{support}(\alpha)\setminus D$, which we  will denote by $v$. In case the other two vertices of $T$ belong to $(\mathsf{support}(\alpha)\setminus D)\cup M$, then the proof is complete. Thus, suppose that $T$ contains at least one vertex $w\in V(G)\setminus (\mathsf{support}(\alpha)\cup M)$. Because $v$ $M$-fits $<_\alpha$ (by Lemma \ref{lem:createModulesPack}), the third vertex of $T$ cannot also belong to $V(G)\setminus (\mathsf{support}(\alpha)\cup M)$, as otherwise $G-((\mathsf{support}(\alpha)\cup M)\setminus\{v\})$ contains a triangle (which contradicts that $v$ $M$-fits $<_\alpha$ due to Proposition \ref{prop:linOrder}). So, the third vertex, which we denote by $u$, belongs to $(\mathsf{support}(\alpha)\setminus D)\cup M$. We suppose that $v$ is of position $0$ or $r<_\alpha w$, as the proof for the other case, where $v$ is of position $r'\geq_\alpha w$, is symmetric. Then, by the definition of position, $(v,w)\in E(G)$. So, because $T$ is a triangle, this implies that $(w,u),(u,v)\in E(G)$. Now, because $u$ $M$-fits $<_\alpha$ (by Lemma \ref{lem:createModulesPack}), having the arc $(w,u)\in E(G)$ implies that $u$ is of position $r'\geq_\alpha w$. This completes the proof.
\end{proof}

We now argue that $|D|$ is only a $\delta$-fraction of the optimum, and hence it is not ``costly'' to seek only solutions that contain $D$. We remark that as we will apply another (non-strict) lossy rule later, we will need to call {\sf Marking} with $\delta<\epsilon$.

\begin{lemma}\label{lem:fvstDSmall}
Given $0<\delta<1$, a tournament $G$ and an optimal solution $\alpha$ to the {\sc $3$-Hitting Set} instance corresponding to $G$, let $\{\mathsf{mark}(v)\}|_{v\in\alpha^{-1}(1)},M,D$ be the output of {\sf Marking}$(\delta,G,\alpha)$. Let $S^\star$ be a solution to $I$. Then, $|D\setminus S^\star|\leq\delta|S^\star|$.
\end{lemma}

\begin{proof}
Consider some vertex $v\in D$. Notice that $v$ together with any edge in $\nu_v$ form a triangle in $G$. Thus, if $v\notin S^\star$, then from every edge in $\nu_v$, at least one vertex must belong to $S^\star$. As $\nu_v$ is a matching, and its size is $\frac{1}{\delta}$, this means that $S^\star$ had to contain at least $\frac{1}{\delta}$ vertices from $\mathsf{mark}(v)$. 
As the sets assigned by $\mathsf{mark}$ are pairwise disjoint, we have that $|D\setminus S^\star|$ can be of size at most $\delta\mathsf{opt}(I)$.
\end{proof}

Intuitively, the arguments above naturally give rise to a reduction rule that deletes $D$. This will be part of our merged rule given later on.

In order to shrink the size of modules, we will need another marking procedure.

\begin{definition}\label{def:markFVST2}
Given $0<\delta,\delta'<1$, a tournament $G$ and an optimal solution $\alpha$ to the {\sc $3$-Hitting Set} instance corresponding to $G$, let $\{\mathsf{mark}(v)\}|_{v\in\alpha^{-1}(1)},M,D$ be the output of {\sf Marking}$(\delta,G,\alpha)$. Then, {\sf ExtraMarking}$(\delta',G,\alpha,M)$ is defined as follows.
\begin{enumerate}
\item For every vertex $v\in \mathsf{support}(\alpha)$, initialize {\sf backw}$(v)=\emptyset$ and {\sf forw}$(v)=\emptyset$.
\item For every vertex $v\in\mathsf{support}(\alpha)\setminus D$:
	\begin{enumerate}
	\item Let $p$ be the position of $v$. 
	
	\item Let $\rho^\mathsf{backw}_v=\{u\in V(G)\setminus(\mathsf{support}(\alpha)\cup M\cup \bigcup_{r\in\mathsf{support}(\alpha)}(\mathsf{forw}(r)\cup\mathsf{backw}(r))): u\leq_\alpha p\}$.  If $|\rho^\mathsf{backw}_v|>\frac{1}{\delta'}$, then let $\mathsf{backw}(v)$ be the subset of the $\frac{1}{\delta'}$ largest (according to $<_\alpha$) vertices in $\rho^\mathsf{backw}_v$, and otherwise let $\mathsf{backw}(v)=\rho^\mathsf{backw}_v$.  

	\item Let $\rho^\mathsf{forw}_v=\{u\in V(G)\setminus(\mathsf{support}(\alpha)\cup M\cup \bigcup_{r\in\mathsf{support}(\alpha)}(\mathsf{forw}(r)\cup\mathsf{backw}(r))): p<_\alpha u\}$.  If $|\rho^\mathsf{forw}_v|>\frac{1}{\delta'}$, then let $\mathsf{forw}(v)$ be the subset of the $\frac{1}{\delta'}$ smallest (according to $<_\alpha$) vertices in $\rho^\mathsf{forw}_v$, and otherwise let $\mathsf{forw}(v)=\rho^\mathsf{forw}_v$.  
	\end{enumerate}
\item For every $v\in \mathsf{support}(\alpha)$, output $\mathsf{backw}(v),\mathsf{forw}(v)$, and $\widehat{M}=\bigcup_{v\in\mathsf{support}(\alpha)\setminus D}(\mathsf{backw}(v)\cup\mathsf{forw}(v))$.
\end{enumerate}
\end{definition}

The main utility of this marking scheme is given by the following lemma. 

\begin{lemma}\label{lem:fvstExtraMarkSafe}
For $0<\delta,\delta'<1$, a tournament $G$ and an optimal solution $\alpha$ to the {\sc $3$-Hitting Set} instance corresponding to $G$, let $\{\mathsf{mark}(v)\}|_{v\in\alpha^{-1}(1)},M,D$ be the output~of~{\sf Marking}$(\delta,G,$ $\alpha)$, and $\{b_v,f_v,\mathsf{backw}(v),\mathsf{forw}(v)\}|_{v\in\mathsf{support}(\alpha)},\widehat{M}$ be the output of  {\sf ExtraMarking}$(\delta',G,\alpha,M)$.
Let $v\in\mathsf{support}(\alpha)\setminus D,u\in(\mathsf{support}(\alpha)\setminus D)\cup M, w\in V(G)\setminus(\mathsf{support}(\alpha)\cup M\cup\widehat{M})$ such that $G[\{v,u,w\}]$ is a triangle. Then, the following conditions hold.
\begin{itemize}
\item If $(v,u)\in E(G)$, then $|\mathsf{backw}(v)|=\frac{1}{\delta'}$ and for every $r\in \mathsf{backw}(v)$, $G[\{v,u,r\}]$ is a triangle.
\item Otherwise (when $(u,v)\in E(G)$), then $|\mathsf{forw}(v)|=\frac{1}{\delta'}$ and for every $r\in \mathsf{forw}(v)$, $G[\{v,u,r\}]$ is a triangle.
\end{itemize}
\end{lemma}

\begin{proof}
We only give a proof for the case where $(v,u)\in E(G)$, as the proof for the case where $(u,v)\in E(G)$ is symmetric. Then, $(w,v)\in E(G)$. So, Lemma \ref{lem:traignleCharacterization} implies that $u$ is of position $0$ or $p'<_\alpha w$, and $v$ is of position $p\geq_\alpha w$. Thus, $w\in\rho^\mathsf{backw}(v)$.  Having $w\in V(G)\setminus(\mathsf{support}(\alpha)\cup M\cup\widehat{M})$ also means that $w\notin\mathsf{backw}(v)$, and therefore necessarily $|\mathsf{backw}(v)|=\frac{1}{\delta'}$. Now, consider some $r\in \mathsf{backw}(v)$. Because $w\in\rho^\mathsf{backw}(v)$ but $w\notin\mathsf{backw}(v)$, this means that $w<_\alpha r$ (because we insert the largest vertices from $\rho^\mathsf{backw}(v)$ into $\mathsf{backw}(v)$). Hence, since $u$ is of position $0$ or $p'<_\alpha w$, we have that $(u,r)\in E(G)$. Further, by the definition of $\rho^\mathsf{backw}(v)$, we know that $r\leq p$, and therefore $(r,v)\in E(G)$. Thus, indeed $G[\{v,u,r\}]$ is a triangle. This completes the proof.
\end{proof}

We now argue that if either all vertices in $\mathsf{backw}(v)$ are deleted or all vertices in $\mathsf{forw}(v)$ are deleted (or both),  then it is not ``costly'' to seek only solutions that delete $v$ as well.

\begin{lemma}\label{lem:fvstDhatSmall}
Given $0<\delta,\delta'<1$, a tournament $G$ and an optimal solution $\alpha$ to the {\sc $3$-Hitting Set} instance corresponding to $G$, let $\{\mathsf{mark}(v)\}|_{v\in\alpha^{-1}(1)},M,D$ be the output of {\sf Marking}$(\delta,G,\alpha)$, and $\{\mathsf{backw}(v),\mathsf{forw}(v)\}|_{v\in\mathsf{support}(\alpha)},\widehat{M}$ be the output of  {\sf ExtraMarking}$(\delta',G,$ $\alpha,M)$. Let $S'$ be a solution to $G'=G-(D\cup X)$ for $X=V(G)\setminus(\mathsf{support}(\alpha)\cup M\cup\widehat{M})$. Let $Y=\{v\in\mathsf{support}(\alpha)\setminus D: |\mathsf{backw}(v)|=\frac{1}{\delta'},\mathsf{backw}(v)\subseteq S'\}\cup \{v\in\mathsf{support}(\alpha)\setminus D: |\mathsf{forw}(v)|=\frac{1}{\delta'},\mathsf{forw}(v)\subseteq S'\}$. Then, $|Y|\leq\delta'|S'|$.
\end{lemma}

\begin{proof}
Because the collection of the sets $\mathsf{backw}(v)$ and $\mathsf{forw}(v)$ taken over all vertices $v\in\mathsf{support}(\alpha)\setminus D$ are pairwise disjoint, $S'$ can contain at most $\delta'|S'|$ such sets of size $\frac{1}{\delta'}$. As $|Y|$ is precisely the number of such sets of size $\frac{1}{\delta'}$ that $S'$ contains, the lemma follows.
\end{proof}

\begin{definition}\label{def:fvstRule}
Given $0<\delta,\delta'<1$, the {\em {\sc FVST}$(\delta,\delta')$ reduction rule} is defined as follows:
\begin{itemize}
\item {\bf reduce:} Let $I=G$ be an instance of {\sc Feedback Vertex Set in Tournaments}. Use the algorithm in Proposition \ref{prop:solveLP} to compute an optimal solution $\alpha$ to the classic LP corresponding to it (Definitions \ref{def:HSLP} and \ref{def:FVSTtoHS}). Let $\{\mathsf{mark}(v)\}|_{v\in\alpha^{-1}(1)},M,D$ be the output of {\sf Marking}$(\epsilon,G,\alpha)$. Let $\{\mathsf{backw}(v),\mathsf{forw}(v)\}|_{v\in\mathsf{support}(\alpha)},\widehat{M}$ be the output of  {\sf ExtraMarking}$(\delta',$ $G,\alpha,M)$.

Output $I'=G'$ where $G'=G-(D\cup X)$ for $X=V(G)\setminus(\mathsf{support}(\alpha)\cup M\cup\widehat{M})$.

\item {\bf lift:} Given $I,I'$ and a solution $S'$ to $I'$, output  $S=S'\cup D\cup Y$ where $Y=\{v\in\mathsf{support}(\alpha)\setminus D: |\mathsf{backw}(v)|=\frac{1}{\delta'},\mathsf{backw}(v)\subseteq S'\}\cup \{v\in\mathsf{support}(\alpha)\setminus D: |\mathsf{forw}(v)|=\frac{1}{\delta'},\mathsf{forw}(v)\subseteq S'\}$.
\end{itemize}
\end{definition}

Just like Lemma \ref{lem:cvdSupport} in Section \ref{sec:P3}, here also we present a simple lemma that will help us derive a tighter bound on the number of vertices in the output graph. Since the proof follows the exact same arguments as the proof of Lemma \ref{lem:cvdSupport}, it is omitted. 

\begin{lemma}\label{lem:fvstSupport}
Let $I=G$ be an instance of {\sc Feedback Vertex Set in Tournaments}, and let $\alpha$ be a solution to the {\sc $3$-Hitting Set} instance corresponding to $G$. Then, $|\mathsf{support}(\alpha)|\leq 3\mathsf{frac}(I)-2|\alpha^{-1}(1)|$.
\end{lemma}

Based on Lemmata \ref{lem:traignleCharacterization}, \ref{lem:fvstDSmall}, \ref{lem:fvstExtraMarkSafe}, \ref{lem:fvstDhatSmall} and \ref{lem:fvstSupport}, we are now ready to prove the main theorem of this subsection.

\begin{theorem}\label{thm:fvst}
Let $0<\epsilon<1$. The {\sc Feedback Vertex Set in Tournaments} problem, parameterized by the fractional optimum of the classic LP, admits a $(1+\epsilon)$-approximate $(13+\frac{9}{\epsilon})\mathsf{frac}(I)$-vertex kernel.
\end{theorem}

\begin{proof}
Our lossy kernelization algorithm consists only of the {\sc FVST$(\delta,\delta')$} reduction rule where $\delta=\frac{\epsilon}{3}-\frac{2\epsilon^2}{9},\delta'=\frac{2\epsilon}{3}$. Clearly, it runs in polynomial time. 

First, we consider the number of vertices in the output graph $G'$ of {\bf reduce}. By Lemma~\ref{lem:fvstSupport}, $|\mathsf{support}(\alpha)|\leq 3\mathsf{frac}(I)-2|\alpha^{-1}(1)|$ (I). Moreover, $|M|=|\bigcup_{v\in\alpha^{-1}(1)}\mathsf{mark}(v)| = 2|\bigcup_{v\in\alpha^{-1}(1)}\nu_v|\leq \frac{2}{\delta}|\alpha^{-1}(1)|$ (II). Additionally, $|\widehat{M}|=|\bigcup_{v\in\mathsf{support}(\alpha)\setminus D}(\mathsf{backw}(v)\cup\mathsf{forw}(v))|\leq \frac{2}{\delta'}|\mathsf{support}(\alpha)|$ (III). As $V(G') \subseteq \mathsf{support}(\alpha)\cup M\cup \widehat{M}$ (more precisely, $V(G') = (\mathsf{support}(\alpha)\setminus D)\cup M\cup \widehat{M}$), we have that
\[\begin{array}{lll}
\smallskip
|V(G')| & \leq |\mathsf{support}(\alpha)| + |M|+|\widehat{M}| & [\mbox{Last sentence}]\\\

\smallskip
&\leq |\mathsf{support}(\alpha)| + \frac{2}{\delta}|\alpha^{-1}(1)|+\frac{2}{\delta'}|\mathsf{support}(\alpha)|\\

\smallskip
&= (1+\frac{2}{\delta'})|\mathsf{support}(\alpha)| + \frac{2}{\delta}|\alpha^{-1}(1)|& [\mbox{(I)}]\\

\smallskip
&\leq (1+\frac{2}{\delta'})(3\mathsf{frac}(I)-2|\alpha^{-1}(1)|) + \frac{2}{\delta}|\alpha^{-1}(1)|& [\mbox{(II)+(III)}]\\

\smallskip
&\leq 3(1+\frac{2}{\delta'})\mathsf{frac}(I) + 2(\frac{1}{\delta}-\frac{2}{\delta'}-1)|\alpha^{-1}(1)|\\

\smallskip
&= 3(1+\frac{3}{\epsilon})\mathsf{frac}(I) +2(\frac{9}{\epsilon(3-2\epsilon)}-\frac{3}{\epsilon}-1)|\alpha^{-1}(1)|& [\mbox{Substitute $\delta$ and $\delta'$}]\\

\smallskip
&= 3(1+\frac{3}{\epsilon})\mathsf{frac}(I) +2(\frac{6}{3-2\epsilon}-1)|\alpha^{-1}(1)|\\

\smallskip
&\leq 3(1+\frac{3}{\epsilon})\mathsf{frac}(I) +10|\alpha^{-1}(1)|& [\epsilon<1]\\

& \leq (13+\frac{9}{\epsilon})\mathsf{frac}(I) & [|\alpha^{-1}(1)|\leq \mathsf{frac}(I)].
\end{array}\]

We turn to prove that {\bf lift} returns a solution having the desired approximation ratio. To this end, suppose that it is given $I,I',S'$ where $S'$ is a solution to $I'$. We first show that $S=S'\cup D\cup Y$ is a solution to $I$. Targeting a contradiction, suppose that this is false, and hence there exists a triangle $T$ in $G-S$. As $D\subseteq S$, this triangle also exists in $G-D$, and hence by Lemma \ref{lem:traignleCharacterization}, $T$ consists of either 
\begin{enumerate}
\item three vertices of $(\mathsf{support}(\alpha)\setminus D)\cup M$, or
\item a vertex $v\in\mathsf{support}(\alpha)\setminus D$, a vertex $u\in(\mathsf{support}(\alpha)\setminus D)\cup M$ and a vertex $w\in V(G)\setminus (\mathsf{support}(\alpha)\cup M)$ such that either {\em (i)} $(u,v)\in E(G)$, $v$ is of position $0$ or $r<_\alpha w$, and $u$ is of position $r'\geq_\alpha w$, or {\em (ii)} $(v,u)\in E(G)$, $u$ is of position $0$ or $r<_\alpha w$, and $v$ is of position $r'\geq_\alpha w$.
\end{enumerate}
Since $S'$ is a solution to $I'$, $T$ must consists of at least one vertex from $V(G)\setminus V(G')=X=V(G)\setminus (\mathsf{support}(\alpha)\cup M\cup\widehat{M})$, and therefore the first case is impossible. Moreover, this implies that in the second case, $w\in X$. We only consider the case where $(v,u)\in E(G)$, as the proof for the other case (when $(u,v)\in E(G)$) follows symmetric arguments. So, $T=G[\{v,u,w\}]$ where $v\in\mathsf{support}(\alpha)\setminus D$, $u\in(\mathsf{support}(\alpha)\setminus D)\cup M$, $w\in V(G)\setminus (\mathsf{support}(\alpha)\cup M\cup\widehat{M})$, $(v,u)\in E(G)$, $u$ is of position $0$ or $r<_\alpha w$, and $v$ is of position $r'\geq_\alpha w$. By Lemma \ref{lem:fvstExtraMarkSafe}, this means that $|\mathsf{backw}(v)|=\frac{1}{\delta'}$ and for every $r\in \mathsf{backw}(v)$, $G[\{v,u,r\}]$ is a triangle. As $S'$ is a solution to $I'$ that excludes $u$ and $v$, and as for every $r\in \mathsf{backw}(v)$, $G[\{v,u,r\}]$ exists in $G'$, we have that $\mathsf{backw}(v)\subseteq S'$. However, this implies that $v\in Y$, and hence $v\in S$, so $T$ cannot exist in $G-S$. As we have reached a contradiction, $S$ is indeed a solution to $I$. 

It remains to consider the approximation ratio of $S$. To this end, first note that  $|S|\leq |S'|+|D|+|Y|$. So, by Lemma \ref{lem:fvstDhatSmall}, {\em (i)} $|S|\leq (1+\delta')|S'|+|D|$. On the other hand, let $S^\star$ be an optimal solution to $I$. Observe that, as $G'$ is a subgraph of $G$, $S^\star\cap V(G')$ is a solution to $I'$. So, $\mathsf{opt}(I')\leq |S^\star\cap V(G')|$. Further, $S^\star\cap V(G)\subseteq S^\star\setminus D$, and by Lemma \ref{lem:fvstDSmall}, $|D\setminus S^\star|\leq \delta|S^\star|$. Thus, $|S^\star\cap V(G')|\leq |S^\star\setminus D| = |S^\star|-|S^\star\cap D|=|S^\star|-(|D|-|D\setminus S^\star|)\leq (1+\delta)|S^\star|-|D|$, which means that {\em (ii)} $\mathsf{opt}(I')\leq (1+\delta)\mathsf{opt}(I)-|D|$. 
Notice that $(1+\delta)(1+\delta')=(1+\frac{\epsilon}{3}-\frac{2\epsilon^2}{9})(1+\frac{2\epsilon}{3})\leq(1+\epsilon)$. Then, from {\em (i)} and {\em (ii)}, we conclude that
\[\begin{array}{ll}
\smallskip
\displaystyle{\frac{|S|}{\mathsf{opt}(I)}}&\leq \displaystyle{(1+\delta)\frac{(1+\delta')|S'|+|D|}{\mathsf{opt}(I')+|D|}}\\

\smallskip
&\leq \displaystyle{(1+\delta)(1+\delta')\frac{|S'|+|D|}{\mathsf{opt}(I')+|D|}}\\

\smallskip
&\leq \displaystyle{(1+\epsilon)\frac{|S'|+|D|}{\mathsf{opt}(I')+|D|}}\\

&\leq \displaystyle{(1+\epsilon)\max\{\frac{|S'|}{\mathsf{opt}(I')}, \frac{|D|}{|D|}\}}\\

&= \displaystyle{(1+\epsilon)\frac{|S'|}{\mathsf{opt}(I')}.}
\end{array}\]
Here, the last inequality follows from Proposition \ref{prop:numbers}. This completes the proof.
\end{proof}

\begin{corollary}
Let $0<\epsilon<1$. The {\sc Feedback Vertex Set in Tournaments} problem, parameterized by the optimum, admits a $(1+\epsilon)$-approximate $(13+\frac{9}{\epsilon})\mathsf{opt}$-vertex kernel.
\end{corollary}

Due to Lemma \ref{lem:lossyKerOptToK}, we also have the following corollary of Theorem \ref{thm:fvst}.

\begin{corollary}
Let $0<\epsilon<1$. The {\sc Feedback Vertex Set in Tournaments} problem, parameterized by a bound $k$ on the solution size, admits a $(1+\epsilon)$-approximate $\frac{13+\frac{9}{\epsilon}}{1+\epsilon}k$-vertex kernel.
\end{corollary}

%

\section{Conclusion}\label{sec:conclusion}
In this paper, we presented positive results on the kernelization complexity of \probdHS, as well as for its special cases  \cvdfull and  \fvstfull. First, we proved that if we allow the kernelization to be {\em lossy} with a qualitatively better loss than the best possible approximation ratio of polynomial time approximation algorithms, then one can obtain kernels where the number of elements is linear for every fixed $d$.  Further, we extended the notion of lossy kernelization algorithms to {\em lossy kernelization protocols} and, then, presented our main result: For any $\epsilon>0$,  {\sc $d$-Hitting Set} admits a (randomized) pure $(d-\delta)$-approximate kernelization protocol of call size $\OO(k^{1+\epsilon})$. Here, the number of rounds and $\delta$ are fixed constants (that depend only on $d$ and $\epsilon$). Finally, we complemented the aforementioned results as follows: for the  special cases of   \probTHS, namely, \cvdfull and  \fvstfull, we showed that for any $0<\epsilon<1$, they admits a $(1+\epsilon)$-approximate $\OO(\frac{1}{\epsilon}\cdot\mathsf{opt})$-vertex kernel.

We conclude the paper with a few interesting open problems.
\begin{enumerate}
\item Does \dhsfull{d} admit a kernel with $f(d)\cdot k^{d-1-\epsilon}$ elements for some fixed $\epsilon>0$, or, even, with just $f(d)\cdot k$ elements? 
\item Does   \dhsfull{d} admit a $(1+\epsilon)$-approximate $\OO(f(\epsilon) \cdot  k)$-element kernel (or protocol)? 
\item Does   \dhsfull{d} admit a $(1+\epsilon)$-approximate $\OO(f(\epsilon) \cdot  k)$-bits kernel (or protocol)?
\item Do \fvstfull\  and {\sc Cluster Vertex Deletion} admit linear vertex kernels? 
\item Are lossy kernelization protocols ``more powerful'' than lossy kernelization algorithms?
\end{enumerate}

\bibliographystyle{plainurl}
\bibliography{book_kernels_fvf,Refs,references,reference-kernel,metakernels_extended}

\newpage
\appendix

\section{Problem Definitions}\label{app:problems}

\noindent{\bf Vertex Cover (VC).} Given a graph $G$, compute a minimum-sized vertex cover $S$ of $G$, that is, a subset $S\subseteq V(G)$ such that every edge in $G$ is incident to at least one vertex in $S$.

\medskip
\noindent{\bf $d$-Hitting Set ($d$-HS).} Given a universe $U$ and a family of sets ${\cal F}\subseteq 2^U$ where each set in $\cal F$ has size $d$, compute a minimum-sized hitting set $S$ of $\cal F$, that is, a subset $S\subseteq U$ such that every set in $\cal F$ has non-empty intersection with $S$. 

Note that {\sc Vertex Cover} is equivalent to {\sc $2$-Hitting Set}.

\medskip
\noindent{\bf Cluster Vertex Deletion (CVD).} Given a graph $G$, compute a minimum-sized subset $S\subseteq V(G)$ such that $G-S$ is a cluster graph.


\medskip
\noindent{\bf Feedback Vertex Set in Tournaments (FVST).} Given a tournament $G$, compute a minimum-sized subset $S\subseteq V(G)$ such that $G-S$ is acyclic.


\end{document}